\documentclass[final,10pt,twocolumn]{autart} 

\usepackage{bbold}
\usepackage{mathbbol}
\usepackage{amssymb}

\usepackage{graphicx}
\usepackage[table,pdftex,hyperref,svgnames]{xcolor}
\usepackage{url,doi,hyperref}
\hypersetup{colorlinks=true, linkcolor=blue, breaklinks=true, urlcolor=blue}

\usepackage{flushend}
\usepackage{blindtext, graphicx}
\usepackage[utf8]{inputenc}
\usepackage{enumerate}
\usepackage{enumitem}
\usepackage{tabulary}
\usepackage{booktabs}

\graphicspath{{./images/}{./fig}}

\usepackage[cmex10]{amsmath}

\let\theoremstyle\relax

\usepackage{amsthm}

\usepackage{amsfonts}
\usepackage{bm}

\usepackage[caption=false,font=footnotesize]{subfig}

\usepackage{lipsum}
\usepackage{mathtools}
\usepackage{cuted}

\usepackage{mathtools,bbm}

\usepackage[prependcaption,colorinlistoftodos]{todonotes}

\newcommand\oprocendsymbol{\hbox{$\square$}}
\newcommand\oprocend{\relax\ifmmode\else\unskip\hfill\fi\oprocendsymbol}

\newtheorem{theorem}{Theorem}[section]
\newtheorem{proposition}[theorem]{Proposition}

\newtheorem{definition}[theorem]{Definition}
\newtheorem{lemma}[theorem]{Lemma}

\newtheorem{remark}[theorem]{Remark}

\def\real{\mathbb{R}} 

\newcommand{\QHbM}{\subscr{Q}{HbM}}
\newcommand{\fHbM}{\subscr{f}{HbM}}
\renewcommand{\QHbM}{\subscr{Q}{homophily}}
\renewcommand{\fHbM}{\subscr{f}{homophily}}

\newcommand{\QIbM}{\subscr{Q}{IbM}}
\newcommand{\fIbM}{\subscr{f}{IbM}}
\renewcommand{\QIbM}{\subscr{Q}{influence}}
\renewcommand{\fIbM}{\subscr{f}{influence}}

\newcommand{\Snz}{\subscr{\mathcal{S}}{nz-row}}
\newcommand{\Sssymm}{\subscr{\mathcal{S}}{s-symm}^{\operatorname{+}}}
\newcommand{\Srsymm}{\subscr{\mathcal{S}}{rs-symm}^{\operatorname{+}}}

\newcommand{\xmin}{\subscr{x}{min}}
\newcommand{\xmax}{\subscr{x}{max}}
\newcommand{\ave}{\operatorname{ave}}

\newcommand\norm[1]{\left\lVert#1\right\rVert}
\newcommand\maxnorm[1]{\subscr{\left\lvert#1\right\rvert}{max}}
\newcommand\minnorm[1]{\subscr{\left\lvert#1\right\rvert}{min}}

\DeclareSymbolFont{bbold}{U}{bbold}{m}{n}
\DeclareSymbolFontAlphabet{\mathbbold}{bbold}
\newcommand{\vect}[1]{\mathbbold{#1}}

\newcommand{\sign}{\operatorname{sign}}
\newcommand{\diag}{\operatorname{diag}}

\newcommand{\until}[1]{\{1,\dots, #1\}}
\newcommand{\subscr}[2]{#1_{\textup{#2}}}

\newcommand{\setdef}[2]{\{#1 \; | \; #2\}}

\usepackage{enumitem}
\renewcommand{\theenumi}{(\roman{enumi})}

\graphicspath{ {images/} }

\newenvironment{example}{\textit{Example:}}

\renewcommand{\baselinestretch}{1}

\DeclareUnicodeCharacter{2010}{-} 

\begin{document}
\renewcommand{\baselinestretch}{1}
\setlength\parindent{1em}

\begin{frontmatter}

\title{Dynamic Social Balance and Convergent Appraisals \\ via Homophily and Influence Mechanisms}
\thanks[footnoteinfo]{This material is
  based upon work supported by, or in part by, the U.S.\ Army
  Research Laboratory and the U.S.\ Army Research Office under grant
  number W911NF-15-1-0577.}
  
\author[IfA]{Wenjun Mei}\ead{meiwenjunbd@gmail.com}, \author[ccdc]{Pedro Cisneros-Velarde}\ead{pacisne@gmail.com}, \author[cas]{Ge Chen}\ead{chenge@amss.ac.cn}, \author[ccdc,soc]{Noah E. Friedkin}\ead{friedkin@soc.ucsb.edu}, \author[ccdc]{Francesco Bullo}\ead{bullo@engineering.ucsb.edu}

\address[IfA]{Automatic Control Laboratory, ETH Zurich, Switzerland}
\address[ccdc]{Center of Control, Dynamical Systems and Computation, University of California, Santa Barbara, USA}
\address[soc]{Department of Sociology, University of California, Santa Barbara, USA}
\address[cas]{Academy of Mathematics and Systems Science, Chinese Academy of Sciences, Beijing, China}


\begin{abstract}
Social balance theory describes allowable and forbidden configurations of the topologies of signed directed social appraisal networks. In this paper, we propose two discrete-time dynamical systems that explain how an appraisal network converges to social balance from an initially unbalanced configuration. These two models are based on two different socio-psychological mechanisms respectively: the homophily mechanism and the influence mechanism. Our main theoretical contribution is a comprehensive analysis for both models in three steps. First, we establish the well-posedness and bounded evolution of the interpersonal appraisals. Second, we fully characterize the set of equilibrium points; for both models, each equilibrium network is composed of an arbitrary number of complete subgraphs satisfying structural balance.   Third, we establish the equivalence among three distinct properties: non-vanishing appraisals, convergence to all-to-all appraisal networks, and finite-time achievement of social balance. In addition to theoretical analysis, Monte Carlo validations illustrate how the non-vanishing appraisal condition holds for generic initial conditions in both models.  Moreover, a numerical comparison between the two models indicates that the homophily-based model might be a more universal explanation for the emergence of social balance. Finally, adopting the homophily-based model, we present numerical results on the mediation and globalization of local conflicts, the competition for allies, and the asymptotic formation of a single versus two factions.
\end{abstract}

\end{frontmatter}

\section{Introduction}

\paragraph*{Motivation and problem description}

Social systems involving friendly/antagonistic relationships are often modeled as signed networks.  \emph{Social
  balance} (also referred to as \emph{structural balance}) theory,
which originated from several seminal works by
Heider~\cite{FH:44,FH:46}, characterizes the stable configurations of
signed social networks, summarized as the famous Heider's axioms: ``Friends' friends are friends; Friends' enemies are enemies; Enemies' friends are enemies; Enemies
 enemies are friends.'' 
Empirical studies for both large-scale networks ~\cite{JL-DH-JK:10,GF-GI-CA:11} and small groups ~\cite{FH:61,MGK:64,HFT:70} indicate that 
social balance is a type of stable configurations frequently observed in real social networks. Dynamic social balance theory, aiming to explain how an initially unbalanced network evolves to a balanced state, has recently attracted much interest. Despite recent progress, it remains a valuable open problem to propose dynamic models that enjoy desirable boundedness and convergence properties. Such models make it possible to further study meaningful predictions and control strategies for the evolution of social networks to balance.

In this paper, we propose two novel discrete-time dynamic social balance models, in which a group of individuals repeatedly update their
interpersonal appraisals via two socio-psychological mechanisms
respectively: the homophily mechanism and the influence mechanism. Loosely speaking, for the homophily mechanism, the interpersonal appraisals of any two individuals in a social group are adjusted based on whether they agree on the appraisals of the group members. For the influence mechanism, each individual assigns influence to others proportionally to her/his appraisal of them. Both mechanisms are well established in the social sciences literature, e.g., see the seminal work by Lazarsfeld and Merton~\cite{PFL-RKM:54}, and the award-winning book by Friedkin and Johnsen~\cite{NEF-ECJ:11}, respectively. For both models, we characterizes their sets of equilibrium and their dynamical behavior. Moroever, we compare these two models via both theoretical analysis and numerical comparisons and give a tentative answer that, compared to the influence mechanism, the homophily mechanism is a more universal explanation for the evolution of appraisal networks to social balance.

\paragraph*{Literature review}

Following the early works by Heider~\cite{FH:44,FH:46}, static social
balance theory has been extensively studied in the last seven decades, including the characterization of the balanced configurations for both complete networks~\cite{FH:53,DC-FH:56} and arbitrary networks~\cite{WdN:99,DE-JK:10}; the measure of the degree of balance~\cite{DC-TCG:66,NMH-RBH-CBDS:69}; the clustering and its relation to balance~\cite{JAD:67,PD-AM:96}; as well as the relevant partitioning algorithms~\cite{PD-DK:01,MK-KSC:12}. Numerous empirical studies have been conducted for different social systems, including social systems at the national level~\cite{FH:61,HBM-RR:85}, at the group level~\cite{MGK:64,CMR-NEF:17}, and at the individual level~\cite{HFT:70,GF-GI-CA:11}. For a comprehensive review we refer to~\cite{XZ-DZ-FYW:15}.

In the last decade, researchers have started to incorporate dynamical
systems into the social balance theory, aiming to explain how a signed
network evolves to a structurally balanced state. Early works include the discrete-time \emph{local triad dynamics} (LTD)~\cite{TA-PLK-SR:05} and \emph{constrained triad dynamics}~\cite{TA-PLK-SR:06}. These models suffer from the existence of unbalanced equilibria, i.e., the \emph{jammed states}. Other works based on network games are proposed by van de Rijt~\cite{AvdR:11} and Malekzadeh et al.~\cite{MM-MF-PJK-HRR-MAS:11}. In all the aforementioned models, the link weights in the signed networks only take values from the set $\{-1,0,1\}$.

Our models are closely related to the continuous-time dynamic social balance models~\cite{KK-PG-PG:05,SAM-JK-RDK-SHS:11,VAT-PVD-PDL:13}, in which the link weights can take arbitrary real values. The model proposed by Ku{\l}akowski et al.~\cite{KK-PG-PG:05} is based on an influence-like mechanism. Theoretical analysis by Marvel et al.~\cite{SAM-JK-RDK-SHS:11} reveals that for symmetric initial conditions, the probability of achieving social balance in finite time tends to 1 as the network size tends to infinity. Traag et al.~\cite{VAT-PVD-PDL:13} extend the set of initial conditions to normal matrices and provide a sufficient condition for finite-time social balance. In~\cite{VAT-PVD-PDL:13}, the authors also propose an alternative continuous-time model based on a homophily mechanism, and prove that the homophily-based model leads to finite-time social balance for generic initial conditions. In addition to theoretical analysis, Ku{\l}akowski et al.~\cite{KK-PG-PG:05} investigate numerically the relation between the formation of factions and the initial appraisal distribution, for the influence-like model. The corresponding results for the homophily-based model is unavailable in previous literature. A non-negligible shortcoming of all the models mentioned above is that, the interpersonal appraisals diverge to infinity in finite time. To remedy this shortcoming, in~\cite{KK-PG-PG:05}, the authors impose a predetermined upper bound of the interpersonal appraisals. As the consequence, the magnitudes of all the appraisals converge to the predetermined upper bound, see the rigorous analysis in~\cite{SW-MC-KK-AB:15}. In addition to those continuous-time models, Jia et al.~\cite{PJ-NEF-FB:13n} propose a discrete-time model, with a generalized notion of social balance and a modified influence mechanism, and establish its convergence to the generalized balance.

\paragraph*{Contributions}

The contribution of this paper are manifold. Our paper is the first to propose two well-behaved discrete-time models that explain the evolution of interpersonal appraisal networks towards the classic Heider's social balance, via the homophily and the influence mechanisms
respectively. Both mechanisms are cast in the language of influence systems; indeed the key novelty is the formulation of appropriate influence matrices such that both models are well-behaved and enjoy the desirable properties of bounded evolution and convergent appraisals.

Regarding the theoretical analysis, we first fully characterizes the two models' respective equilibrium sets, each of which turn out to include all possible balanced configurations in terms of sign pattern. Second, we establish the equivalence relations among the non-vanishing appraisal condition, the convergence of appraisal networks to all-to-all balanced configurations, and the achievement of social balance in finite time.

Numerical study of our both models leads to various insightful results. First, Monte-Carlo validations indicate that the non-vanishing appraisal condition holds for generic initial conditions, while, for the influence-based model, the non-vanishing appraisal condition holds almost surely if the initial appraisals satisfy some generalized notion of symmetry. Second, further simulation results show that, for the influence-based model with generic initial conditions, the probability that the appraisal network converges to social balance monotonically decays to 0 as the network size tends to infinity. Based on this observation we conclude that the homophily-based model might be a more universal explanation than the influence-based model for the evolution to social balance. Third, for the homophily-based model, we numerically investigate its behavior under perturbation when the appraisal network is composed of multiple structurally balanced subnetworks. Such numerical study reveals some insightful and realistic interpretations such as the escalation and mediation of local conflicts. Finally, we study by simulation the effect of the initial appraisal distribution on the formation of factions, i.e., whether an appraisal network converges to two antagonistic factions or an all-friendly network.

The main advantage of our models, compared with the previous continuous-time models~\cite{KK-PG-PG:05,VAT-PVD-PDL:13}, is that our models are well-behaved, in the sense that our models enjoy the desirable property of convergent appraisals, (as opposed to the undesirable property of finite-time divergence). The convergence property makes it possible to characterize the systems' fixed points and their stability, as well as the transition from one equilibrium to another. In our models, the convergent appraisals are due to the introduction of either homophily or interpersonal influence networks, which also provide a connection between the field of dynamic social balance and the field of opinion dynamics with antagonistic interactions, e.g.~\cite{CA:13}. In addition, our models have the desired property that they are invariant under scaling, i.e., if a solution is scaled by a constant, it remains a solution. This feature is particularly important in the modelling of social systems, in which quantities are usually meaningful only in the relative sense. Compared with the model proposed in~\cite{KK-PG-PG:05} with bounded evolution, our models do not rely on any predetermined bound to prevent divergence and the asymptotic appraisals in our models are determined by the initial condition rather than the manually determined bound. Some additional advantages of our models are discussed in Section 5.1.

\paragraph*{Organization}
Section~\ref{notbconc} introduces some notations and basic
concepts. Section~\ref{sectionModelA} and~\ref{sectionModelB} contain the
theoretical analyses of our models. Section~\ref{moredicussions} provides
further discussions and numerical results. Section~\ref{conclSec} gives the
conclusion. An auxiliary lemma is provided in the Appendix. 
  Some proofs are provided in the technical
  report~\cite{WM-PCV-NEF-FB:17f-arxiv} with full details.

\section{Notations and basic concepts}
\label{notbconc} 
\paragraph*{Notations}
Some frequently used notations are defined in Table~\ref{table:notations}. The following sets will be used throughout this paper:
\begin{align}
 \label{eq:def-set-Snz}\Snz  = & \setdef{X\in\real^{n\times{n}}}{ \text{for every } i, X_{i*}\neq \vect{0}_n^{\top}},\\
 \label{eq:def-set-Sssymm}\Sssymm  = & \setdef{X\in\real^{n\times{n}}}{\sign(X)=\sign(X)^\top\\
 \notag    & \qquad \qquad \quad\,\,\,\text{ and }X_{ii} > 0 \text{ for every } i }, \\
 \label{eq:def-set-Sgsymm}\Srsymm  = & \setdef{X\in \Sssymm}{\text{there exists }\gamma \succ \vect{0}_n\\
\notag    &  \qquad \,\,\,\text{such that }\diag(\gamma)X = \big(\diag(\gamma)X\big)^{\top}}.
\end{align}
By definition, $\Srsymm \subset \Sssymm \subset \Snz$. In addition, $\Sssymm$ and $\Srsymm$ are both invariant under permutations. That is, given any $X\in\Sssymm $ (or $X\in\Srsymm$ resp.) and a permutation matrix $P$, we have $PXP^{\top}\in\Sssymm$ (or $PXP^{\top}\in\Srsymm$ resp.).

\begin{table}[htbp]\caption{Notations frequently used in this paper}\label{table:notations}
\begin{center}
\begin{tabular}{r p{6.2cm} }
\toprule
$\vect{1}_n$ ($\vect{0}_n$) & the all-ones (all-zeros) $n\times 1$ vector\\
$\mathbb{R}$ ($\mathbb{Z}_{\ge 0}$) & set of real numbers (non-negative integers)\\
$\succ$ ($\prec$) & entry-wise greater than (less than) \\
$|X|$ & entry-wise absolute value of matrix $X$\\
$\sign(X)$ & entry-wise sign of $X$, i.e., $\sign(X)_{ij}=1$ if $X_{ij}>0$, $\sign(X)_{ij}=-1$ if $X_{ij}<0$, and $\sign(X)_{ij}=0$ if $X_{ij}=0$.\\
$|X|_{\max}$ & the \emph{max norm} of $X$, i.e, $\max_{i,j} |X_{ij}|$ \\
$X_{i*}$ ($X_{*i}$) & the $i$-th row (column) vector of $X$\\
$G(X)$ &  weighted digraph associated with adjacency matrix $X$. We allow negative link weights. That is, if $X_{ij}<0$, then there exists a link in $G(X)$ from $i$ to $j$ with
negative weight $X_{ij}$.\\ 
\bottomrule
\end{tabular}
\end{center}
\end{table}

\paragraph*{Appraisal matrices and social balance}
Given a group of $n$ agents, the interpersonal appraisals are given by the \emph{appraisal matrix} $X\in
\real^{n\times n}$. The sign of $X_{ij}$ determines whether
$i$'s appraisal of $j$ is positive, i.e., $i$ ``likes'' $j$, or
negative, i.e., $i$ ``dislikes'' $j$. The magnitude of $X_{ij}$
represents the intensity of the sentiment. When $X_{ij} = 0$, the
appraisal is one of indifference. The diagonal entry $X_{ii}$
represents agent $i$'s self-appraisal. The weighted digraph
$G(X)$ associated to $X$ as the adjacency matrix is referred to as the
\emph{appraisal network}.
\smallskip

\begin{definition}[Social balance~\cite{FH:53,FH:46}]\label{defSB}
An appraisal network $G(X)$ satisfies social balance, or, equivalently, is structurally balanced, if the appraisal matrix $X$ satisfies the following properties: (S1) $X_{ii}>0$ for any $i\in \until{n}$; (S2) $\sign(X_{ij}) \sign(X_{jk})\sign(X_{ki}) =1$ for any $i,j,k \in \until{n}$.
\end{definition}
According to~\cite{FH:53}, a structurally balanced appraisal network either has only one \emph{faction} in which the interpersonal appraisals are all positive, or is composed of two antagonistic factions such that individuals in the same faction positively appraise each other while all the inter-faction appraisals are negative.

\smallskip
\begin{lemma}[Equivalent conditions for social balance]
  \label{numCond2}
  For any $X\in\real^{n\times n}$ such that all of its entries
  are non-zero, $G(X)$ satisfies social balance if and only if it
  satisfies (S1) in Definition~\ref{defSB} and (S3): $\sign(X_{i*})=\pm\sign(X_{j*})$, for all $i,j\in\until{n}$.
  Moreover, for $G(X)$ satisfying social balance, $X$ is
  sign-symmetric, i.e., $\sign(X)=\sign(X)^{\top}$.
\end{lemma} 

\begin{proof}
Suppose that~(S1) and~(S3) hold. For any $i,j\in \{1,\dots, n\}$, $\sign(X_{i*})=\delta\sign(X_{j*})$, where $\delta$ is either $-1$ or $1$. Therefore, $\sign(X_{ij})\sign(X_{ji})=\delta^2 \sign(X_{jj})\sign(X_{ii})=1$, i.e., $\sign(X_{ij})=\sign(X_{ji})$. Moreover, for any $k$, since $\sign(X_{ij})=\delta \sign(X_{jj})$ and $\sign(X_{jk})=\delta\sign(X_{ik})$, we have 
\begin{align*}
\sign(X_{ij}) & \sign(X_{jk})\sign(X_{ki})\\
& =\delta^2 \sign(X_{jj})\sign(X_{ik})\sign(X_{ki})=1.
\end{align*}
Therefore,~(S1) and~(S3) imply~(S1) and~(S2) in Definition~\ref{defSB}, as well as the sign symmetry of $X$. 

Now suppose~(S1) and~(S2) in Definition~\ref{defSB} hold. The sign symmetry
of $X$ is obtained by letting $k=j$ in~(S2). Moreover, due to the sign
symmetry and~(S2), we obtain
$\sign(X_{ij})\sign(X_{jk})\sign(X_{ik})=1$. Therefore,
$\sign(X_{ik})\sign(X_{jk})$ does not depend on $k$ and is equal to
$\sign(X_{ij})\in\{-1,1\}$. That is, $\sign(X_{i*})=\pm \sign(X_{j*})$ for
any $i$ and $j$. This concludes the proof.
\end{proof}

\section{Homophily-based Model}
\label{sectionModelA}

In this and the next section, we propose and analyze two dynamic social balance models respectively. These two models are distinct in the microscopic individual interaction mechanisms.
\smallskip

\begin{definition}[Homophily-based model]
  Given an initial appraisal matrix
  $X(0)\in{\Sssymm}\subset\real^{n\times{n}}$, the homophily-based
  model is defined by:
  \begin{equation} 
    \label{eq:HbM}
    X(t+1)=\diag(\lvert X(t)\rvert \mathbbm{1}_{n})^{-1}X(t)X^\top(t).
  \end{equation}
\end{definition}

\begin{remark}[Interpretation]
  Equation~\eqref{eq:HbM} updates the appraisals based on what can be
  considered as the \textit{homophily mechanism}. For any
  $i,j\in\until{n}$, agent $i$'s appraisal of agent $j$ at time step
  $t+1$ depends on to what extend they are in agreement with each
  other on the appraisals of all the agents in the group. For any
  $k\in \until{n}$, if $\sign(X_{ik}(t))=\sign(X_{jk}(t))$, then the
  term $X_{ik}(t)X_{jk}(t)$ contributes positively to $X_{ij}(t+1)$,
  and vice versa. The matrix $W(X(t))=\diag(|X(t)|\vect{1}_n)^{-1}X(t)$ can be regarded as the influence matrix constructed from the appraisals through homophily mechanism. Since $X_{ij}(t+1)=\sum_k W_{ik}(t)X_{jk}(t)$, each $|W_{ik}(t)|$ represents how much weight individual $i$ assigns to the agreement on the appraisal of individual $k$. Note that the entry-wise absolute value, i.e., $|W(t)|$, is row-stochastic. Such type of influence matrices has been widely studied in the opinion dynamics with antagonism, see~\cite{JMH:14,WX-MC-KHJ:16,AVP-MC:17}.
\end{remark}

The proposition below presents some useful results on the finite-time
behavior of the homophily-based model.

\smallskip

\begin{proposition}[Invariant set and finite-time behavior of HbM]\label{prop:HbM-finite-time-behav}
Consider the dynamical system~\eqref{eq:HbM} and define $\fHbM (X) =
\diag(\lvert X\rvert \mathbbm{1}_{n})^{-1}XX^\top$.  Pick $X_0\in
\Snz$. The following statements hold:
\begin{enumerate}

\item \label{fact:HbM-inv-set} the map $\fHbM$ is well-defined for any
  $X\in \Snz$ and maps $\Snz$ to $\Sssymm$;

\item \label{fact:HbM-well-defined} the solution $X(t)$,
  $t\in\mathbb{Z}_{\geq0}$, to equation~\eqref{eq:HbM} from initial
  condition $X(0)=X_0$ exists and is unique;

\item \label{fact:HbM-upper-bound} the max norm of any solution $X(t)$
  satisfies
  \begin{equation*}
    \maxnorm{X(t+1)} \leq \maxnorm{X(t)} \leq \maxnorm{X(0)};
  \end{equation*} 
\item \label{fact:HbM-inv-scaling}for any $c>0$, the trajectory $c X(t)$ is the solution to
  equation~\eqref{eq:HbM} from initial condition $X(0)=c X_0$.
\end{enumerate}
\end{proposition}

\begin{proof}
For simplicity, denote $X^+ = \fHbM (X)$. For any $X\in \Snz$, since, for any $i$ and $j$, $X_{ij}^+ = \frac{1}{\lVert X_{i*} \rVert_1}\sum_k X_{ik}X_{jk}$ 
and $\lVert X_{i*} \rVert_1>0$, $\fHbM (X)$ is well-defined. Moreover, 
\begin{align*}
X_{ii}^+ & = \frac{1}{\lVert X_{i*} \rVert_1}\sum_k X_{ik}X_{ik} = \frac{\lVert X_{i*} \rVert_2^2}{\lVert X_{i*}\rVert_1}>0,\quad \text{and }\\
X_{ij}^+ & = \frac{\lVert X_{j*} \rVert_1}{\lVert X_{i*} \rVert_1}X_{ji}^+,\quad \text{for any }i\text{ and }j.
\end{align*}
Therefore, $\fHbM$ maps $\Snz$ to $\Sssymm$. This concludes the proof of statement~\ref{fact:HbM-inv-set}. Statements~\ref{fact:HbM-well-defined} 
is a direct consequence of statement~\ref{fact:HbM-inv-set}, since, for any $t\in \mathbb{Z}_{\ge 0}$, $X(t)\in \Snz$ defines a unique $X(t+1)=\fHbM(X(t))\in \Sssymm$. In addition, 
\begin{align*}
|X_{ij}^+| & \!\le \frac{1}{\lVert X_{i*} \rVert_1}\sum_{k=1}^n \big| X_{ik}X_{jk} \big| \!\le \!\frac{1}{\lVert X_{i*} \rVert_1}\sum_{k=1}^n \!\big| X_{ik}\big| \big| X_{jk} \big|\\
& \le \max_k |X_{jk}|\le \maxnorm{X}
\end{align*}
immediately leads to statement~\ref{fact:HbM-upper-bound}. Finally, statement~\ref{fact:HbM-inv-scaling} is obtained by replacing $X(t)$ with $cX(t)$ on the right-hand side of equation~\eqref{eq:HbM}.
\end{proof}

According to statement (iii) of Proposition~\ref{prop:HbM-finite-time-behav}, for any $a>0$, the set $\Snz\cap [-a,a]^{n\times n}$ is positively invariant under dynamics~\eqref{eq:HbM}. This desired bounded-evolution property makes our model substantially different from some previous models, in which $X(t)$ diverges in finite time~\cite{SAM-JK-RDK-SHS:11,VAT-PVD-PDL:13}.

The theorem below characterizes the set of fixed points of system~\eqref{eq:HbM}, i.e, the steady-state appraisal matrix $X$ satisfying $X = \fHbM(X)$. Fixed points are sociologically interesting because they correspond to the states that can often be observed in the real world.

\smallskip
\begin{theorem}[Fixed points and balance]
  \label{thm:eq-set-QA}
  Consider the dynamical system~\eqref{eq:HbM} in domain $\Snz $. Define
  \begin{equation*}
    \begin{split}
            &\QHbM  \\
            & \text{  }=  \Big{\{}  PYP^\top  \in \Snz  \,\Big|\, P \emph{ is a permutation matrix,}\\
            & \qquad \text{  }Y\emph{ is a block diagonal matrix with blocks of}\\
            & \qquad \text{  }\emph{the form }\alpha bb^\top,\; \alpha > 0, b \in \{ -1,+1 \}^{m},\; m \leq n \Big{\}}.
    \end{split}
  \end{equation*}
  Then 
  \begin{enumerate}
  \item \label{prop:eqSetQA-AllFixedPoints} $\QHbM$ is the set of all the fixed points of~\eqref{eq:HbM},
  \item \label{prop:eqSetQA-StructuralBalance} for any $X\in\QHbM$, $G(X)$ is composed by isolated complete subgraphs that satisfy social balance.
  \end{enumerate}
\end{theorem}

\begin{proof}
We first prove that any $X^*\in \QHbM$ is a fixed point of system~\eqref{eq:HbM}. For any $\alpha>0$ and $b\in \{-1,+1\}^n$, the matrix $Y=\alpha bb^{\top}$ satisfies
\begin{align*}
\fHbM(Y) = \diag(n\alpha \mathbbm{1}_n)^{-1}\alpha^2 bb^{\top}bb^{\top} = \alpha bb^{\top} = Y.
\end{align*} 
This arguments extend to block diagonal matrices $Y$. By the definition of $\fHbM$, for any block diagonal matrix $Y=\diag(Y^{(1)},\dots,Y^{(K)})$, $Y=\fHbM(Y)$ if and only if $Y^{(i)} = \diag(|Y^{(i)}|\mathbbm{1}_n)^{-1} Y^{(i)} {Y^{(i)}}^{\top}$ for any $i$. Therefore, $Y$ is a fixed point of system~\eqref{eq:HbM} if each $Y^{(i)}$ in $Y=\diag(Y^{(1)},\dots,Y^{(K)})$ is a $n_i\times n_i$ matrix of the form $\alpha_i b^{(i)} {b^{(i)}}^{\top}$, with $\alpha_i>0$, $b^{(i)}\in \{-1,+1\}^{n_i}$, and $n_1+\dots+n_K=n$.
Moreover, given any fixed point $Y$, for any permutation matrix $P\in \real^{n\times n}$, 
\begin{align*}
PYP^{\top} & = P\diag(|Y|\mathbbm{1}_n)^{-1}YY^{\top}P^{\top}\\
           & = \diag(|PYP^{\top}|\mathbbm{1}_n)^{-1}(PYP^{\top})(PYP^{\top})^{\top} \\
           & = \fHbM(PYP^{\top}).
\end{align*}
Therefore, any $X^*\in \QHbM$ is a fixed point of~\eqref{eq:HbM}.

Now we prove by induction that $\QHbM$ is the set of all the fixed points of system~\eqref{eq:HbM}. For the trivial case of $n=1$, $\QHbM$ represents the set of all the positive scalars and one can easily check that any positive scalar $X$ is a fixed point of system~\eqref{eq:HbM} with $n=1$. Suppose statement~\ref{prop:eqSetQA-AllFixedPoints} holds for any system with dimension $\tilde{n}<n$. 
For system~\eqref{eq:HbM} with dimension $n$, suppose $X$ is a fixed point, i.e., $X=\fHbM(X)$. For any $i,\,j\in \{1,\dots,n\}$, by comparing the $(i,j)-$th and the $(j,i)-$th equations of $X=\fHbM(X)$, we conclude that $X_{ij}$ and $X_{ji}$ always have the same sign.
In addition, since $X_{ii}=\sum_{k=1}^n X_{ik}^2 \big/ \lVert X_{i*} \rVert_1$, we have $X_{ii}>0$ for any $i$.
Since $X$ is a fixed point of $\fHbM$, we have that, for any $i,j\in \{1,\dots,n\}$,
\begin{align*}
|X_{ij}| & = \frac{1}{\lVert X_{i*} \rVert_1} \Big{|} \sum_k X_{ik}X_{jk} \Big{|}\\
         & \le \frac{1}{\lVert X_{i*} \rVert_1} \sum_k |X_{ik}||X_{jk}|\le \maxnorm{X} .
\end{align*}
Moreover, there exists $(i,j)$ such that $|X_{ij}|=\maxnorm{X}$. For any such $(i,j)$, either of the following two cases hold:

Case 1: $i=j$ and there does not exist $k\neq i$ such that
$|X_{ik}|=\maxnorm{X}$. In this case,
  $|X_{ii}|=\maxnorm{X}$. Since
\begin{align*}
|X_{ii}| & =\frac{1}{\lVert X_{i*} \rVert_1}\Big| \sum_k X_{ik}X_{ik} \Big|\\
& \le \frac{1}{\lVert X_{i*} \rVert_1}\sum_k |X_{ik}||X_{ik}| \le \maxnorm{X},
\end{align*}
in order for $|X_{ii}|=\maxnorm{X}$ to hold, $X_{i*}$ must satisfy
$|X_{ik}|=\maxnorm{X}$, for any $k$ such that $X_{ik}\neq 0$. By the definition
of Case~1, we conclude that there does not exist $k\neq i$ such that
$X_{ik}\neq 0$.  Therefore, there exists a permutation matrix $P$ such
that
\begin{equation*}
PXP^{\top} = 
\begin{bmatrix}
\maxnorm{X} & \vect{0}_{n-1}^{\top} \\
\vect{0}_{n-1} & \tilde{X}_{(n\!-\!1)\times (n\!-\!1)}
\end{bmatrix}.
\end{equation*}
Since $PXP^{\top}$ is also a fixed point of system~\eqref{eq:HbM}, one can check that $\tilde{X}$ satisfies $\tilde{X} = \diag(|\tilde{X}|\vect{1}_n)^{-1}\tilde{X}\tilde{X}^{\top}$.
Therefore, $\tilde{X}$ is a fixed point of system~\eqref{eq:HbM} with dimension $n\!-\!1$. Since we have assumed that statement~\ref{prop:eqSetQA-AllFixedPoints} holds for dimension $\tilde{n}<n$, there exists an $(n\!-\!1)\!\times\! (n\!-\!1)$ permutation matrix $\tilde{P}$ and a block diagonal $\tilde{Y}$, with blocks of the form $\alpha bb^{\top}\!$, where $\alpha>0$, $b\in\{-1,+1\}^{m}$, $m<n\!-\!1$, such that $\tilde{X}=\tilde{P}\tilde{Y}\tilde{P}^{\top}$. Therefore,
\begin{equation*}
X = P^{\top}\!
\begin{bmatrix}
1 & \hspace{-0.1cm}\vect{0}_{n-1}^{\top} \\
\vect{0}_{n-1} & \hspace{-0.1cm}\tilde{P}
\end{bmatrix}
\begin{bmatrix}
\maxnorm{X} & \hspace{-0.05cm}\vect{0}_{n-1}^{\top}\\
\vect{0}_{n-1} & \hspace{-0.05cm}\tilde{Y}
\end{bmatrix}
\begin{bmatrix}
1 & \hspace{-0.1cm}\vect{0}_{n-1}^{\top}\\
\vect{0}_{n-1} & \hspace{-0.1cm}\tilde{P}
\end{bmatrix}^{\top}\!
P.
\end{equation*}
The matrix
$
P^{\top}
\begin{bmatrix}
1 & \vect{0}_{n-1}^{\top} \\
\vect{0}_{n-1} & \tilde{P}
\end{bmatrix}
$
is also a permutation matrix. Therefore $X\in \QHbM$.

Case 2: $j\neq i$ and $|X_{ij}|=\maxnorm{X}$. We first define some notations used in the following proof: For any $k$, let $\theta_k = \setdef{\ell}{X_{k\ell} \neq 0}$ and $|\theta_k|$ be the cardinality of the set $\theta_k$. Note that, since $X=\fHbM(X)\in \Sssymm $, $k$ is always in $\theta_k$ and $X_{kk}>0$. Let $X_{\ell*,\theta_k}\in \real^{1\times |\theta_k|}$ be the $\ell$-th row vector of $X$ with all the $X_{\ell p}$ entries such that $p\notin \theta_k$ removed.

We point out a general result that, for any $k$ and $\ell$, if
\begin{equation*}
|X_{k\ell}| = \frac{1}{\lVert X_{k*} \rVert_1} \Big| \sum_{p=1}^n X_{kp}X_{\ell p} \Big| = \maxnorm{X},
\end{equation*} 
then, for the second equality to hold, $X$ must satisfy that: 1) $\theta_k\subset \theta_l$; 2) $|X_{\ell p}|=\maxnorm{X}$ for any $p\in \theta_k$; 3) $\sign(X_{\ell*,\theta_k})=\pm \sign(X_{k*,\theta_k})$.
Therefore, for the $i,j$ indexes such that $|X_{ij}|=\maxnorm{X}$ and $i\neq j$, we have: $|X_{jk}|=\maxnorm{X}$, for any $k\in \theta_i$; $\theta_i \subset \theta_j$; and $\sign(X_{j*,\theta_i})=\pm\sign(X_{i*,\theta_i})$. Since $i\in \theta_i$ and $X=\fHbM(X)$, we obtain $|\fHbM(X)_{ji}|=|X_{ji}|=\maxnorm{X}$. Therefore, $|\fHbM(X)_{ik}|=|X_{ik}|=\maxnorm{X}$, for any $k\in \theta_j$, and $\theta_j \subset \theta_i$, which in turn leads to $\theta_i=\theta_j$ and $|X_{ik}|=\maxnorm{X}$ for any $k\in \theta_i$. Therefore, for any $k\in \theta_i$, $|\fHbM(X)_{ik}|=\maxnorm{x}$, which implies $|X_{k\ell}|=\maxnorm{X}$ for any $l\in \theta_i$. Since $|\fHbM(X)_{k\ell}|=|X_{k\ell}|$, we further obtain that $\theta_k \subset \theta_l$ and $\sign(X_{k*,\theta_k})=\pm \sign(X_{\ell*,\theta_k})$. Moreover, due to the fact that the indexes $k$ and $l$ are interchangeable, we conclude that, for any $k,l\in \theta_i$: a) $\theta_k=\theta_l=\theta_i$; b) $|X_{k\ell}|=\maxnorm{X}$; c) $\sign(X_{k*})=\pm \sign(X_{\ell*})$.

If $|\theta_i|=n$, let $\alpha=X_{11}$ and $b=\sign(X_{1*})^{\top}$, then we have $X=\alpha bb^{\top}$. If $|\theta_i|<n$, there exists a permutation matrix $P$ such that 
\begin{equation*}
PXP^{\top} = 
\begin{bmatrix}
X^{(\theta_i)} & \vect{0}_{|\theta_i|\times (n-|\theta_i|)} \\
\vect{0}_{(n-|\theta_i|)\times |\theta_i|} & \tilde{X}
\end{bmatrix},
\end{equation*}
where $X^{(\theta_i)}$ is a $|\theta_i|\times |\theta_i|$ matrix. Moreover, $X^{(\theta_i)}=\maxnorm{X} bb^{\top}$, where $b=\sign(X_{i*,\theta_i})^{\top}$. Following the same line of argument for Case 1, we know that $\tilde{X}$ is of the form $\tilde{P}\tilde{Y}\tilde{P}^{\top}$ and thereby $X\in \QHbM$. This concludes the proof for statement~\ref{prop:eqSetQA-AllFixedPoints}.

For any $X^*\in \QHbM$, there exists a permutation matrix $P$ and a block diagonal matrix $Y=\diag(Y^{(1)},\dots,Y^{(K)})$ such that $X^*=PYP^{\top}$. Note that $G(Y)$ has exactly the same topology as $G(X)$, but with the nodes re-indexed. Therefore, we only need to analyze the structure of $G(Y)$. The graph $G(Y)$ is made up of $K$ isolated complete subgraphs and $Y^{(i)} = \alpha_i b^{(i)} {b^{(i)}}^{\top}$ for each such subgraph $G(Y^{(i)})$, where $b^{(i)}=(b^{(i)}_{1},\dots,b^{(i)}_{n_i})^{\top}$. Therefore, according to Lemma~\ref{numCond2}, each subgraph $G(Y^{(i)})$ satisfies social balance.
This concludes the proof for statement~\ref{prop:eqSetQA-StructuralBalance}. 
\end{proof}

\begin{remark}[Social balance with multiple isolated subgraphs]\label{remark:k-block-balance}
An appraisal matrix $X\in \QHbM$ can be a block-diagonal matrix $\diag(X_1,\dots,X_k)$ and thus corresponds to an appraisal network $G(X)$ composed of $k$ isolated subgraphs, each of which satisfies social balance as in Definition~\ref{defSB}. With the notion of social balance extended to graphs with multiple isolated subgraphs, in terms of sign pattern, the set of fixed points $X$ of the homophily-based model~\eqref{eq:HbM} corresponds to exactly the set of all the possible structurally balanced configurations of the appraisal network $G(X)$. Such characterization of fixed points is impossible in the previous continuous-time models~\cite{SAM-JK-RDK-SHS:11,VAT-PVD-PDL:13} since those models diverge in finite time. Moreover, for any $X\in \QHbM$ such that $G(X)$ has $k$ isolated subgraphs, $X$ is a rank-$k$ matrix.
\end{remark}

Before presenting the main results on the convergence of the appraisal matrix $X(t)$ to social balance, we define a property of $X(t)$ as the solution to equation~\ref{eq:HbM}.
\smallskip

\begin{definition}[Non-vanishing appraisal condition]\label{def:lower-bounded-appraisal-condition}
A solution $X(t)$ satisfies the non-vanishing appraisal condition if $\liminf\limits_{t\to\infty}
\min\limits_{i,j}|X_{ij}(t)| > 0$.
\end{definition}
\smallskip

\begin{theorem}[Convergence and social balance in HbM]
\label{theoremMain}
Consider the homophily-based model given by equation~\eqref{eq:HbM}. The following statements hold:
\begin{enumerate}
\item Each element in $\QHbM$ of rank one is a locally stable fixed point of $\fHbM$;\label{statement:Thm-HbM-local-stability}
\item For any $X(0)\in \Snz$, the following three statements are equivalent:
   \begin{enumerate}[label=(\alph*)]
   \item the solution $X(t)$ satisfies the non-vanishing appraisal condition;
   \item there exists $t_0>0$ such that $G(X(t))$ satisfies social balance for all $t\ge t_0$;
   \item there exists $X^*\in \QHbM$ of rank one such that $\lim_{t\to \infty} X(t)=X^*$.
   \end{enumerate}
\end{enumerate}
\end{theorem}

\emph{Proof: }
For simplicity of notations, let $\minnorm{X}=\min_{k,l}|X_{k,l}|$. We start by proving the following two claims. For any given $t_0\ge 0$, if all the entries of $X(t_0)$ are non-zero and $G(X(t_0))$ satisfies social balance, then,
\begin{enumerate}[label={C.\arabic*)}]
\item\label{claim-modelA-main-theorem-1} for any $t\ge t_0$, $G(X(t))$ satisfies social balance and $\sign(X(t))=\sign(X(t_0))$;
\item\label{claim-modelA-main-theorem-2} for any $t\ge t_0$, $\maxnorm{X(t)}$ is non-increasing and $\minnorm{X(t)}$ is non-decreasing.
\end{enumerate}
To prove claim~\ref{claim-modelA-main-theorem-1}, it suffices to prove that $G(X(t_0+1))$ satisfies social balance and $\sign(X(t_0+1))=\sign(X(t_0))$, as the cases for $t\ge t_0+1$ follow by induction. For any $i$ and $j$, since $G(X(t_0))$ satisfies social balance, according to Lemma~\ref{numCond2}, we have $\sign(X_{i*}(t_0))=\pm \sign(X_{j*}(t_0))$. In addition, we have $X_{jj}(t_0)>0$ for any $j$. Therefore,
\begin{align*}
\sign\!\big( \! X_{ij}(t_0\!+\!1)\big) \! & =\! \sign\!\Big(\frac{1}{\lVert X_{i*}(t_0) \rVert_1} \!\sum_{k=1}^n X_{ik}(t_0)X_{jk}(t_0)\Big)\\
&\! =\sign\!\big( X_{ij}(t_0)X_{jj}(t_0) \big)\!=\!\sign\!\big( X_{ij}(t_0)\big),
\end{align*} 
for any $i$ and $j$. This concludes the proof for claim~\ref{claim-modelA-main-theorem-1}. For any $t\ge t_0$, since $G(X(t))$ satisfies social balance, 
\begin{equation*}
|X_{ij}(t+1)| = \frac{1}{\lVert X_{i*}(t) \rVert_1}\sum_{k=1}^n |X_{ik}(t)||X_{jk}(t)|\text{ for any }i,j,
\end{equation*}
we have
$\minnorm{X(t+1)} \ge \minnorm{X(t)} \ge \minnorm{X(t_0)}$ and
$\maxnorm{X(t+1)} \le \maxnorm{X(t)} \le \maxnorm{X(t_0)}$.

Now we prove statement~\ref{statement:Thm-HbM-local-stability}, i.e., each $X^*\in \QHbM$ with rank 1 is locally stable. Let $X^*=\alpha bb^{\top}$, where $\alpha>0$ and $b\in \{-1,+1\}^n$. For any matrix $\Delta\in \real^{n\times n}$ such that $\maxnorm{\Delta}=\zeta<\alpha$, we have $\sign(X^*+\Delta)=\sign(X^*)$. Due to claim~\ref{claim-modelA-main-theorem-1} and~\ref{claim-modelA-main-theorem-2}, we know that, for $X(0)=X^*+\Delta$, $X(t)$ satisfies that, for any $t\ge 0$: (1) $\sign(X(t))=\sign(X(0))=\sign(X^*)$; (2) $\alpha-\zeta \le \minnorm{X(t)}\le \maxnorm{X(t)}\le \alpha+\zeta$. 
Therefore, for any $i$ and $j$, $X_{ij}(t)$ is of the form $\alpha_{ij}(t)\sign(X^*_{ij})$, where $0<\alpha-\zeta \le \alpha_{ij}(t)\le \alpha+\zeta$. We thereby have
\begin{align*}
 \maxnorm{X(t)-X^*}  & = \max_{ij} \big| \alpha_{ij}(t)\sign(X^*_{ij}) - \alpha \sign(X^*_{ij}) \big|\\
& = \max_{ij} | \alpha_{ij}(t) - \alpha |\le \zeta.
\end{align*}
Therefore, for any $\epsilon>0$, there exists $\zeta=\min\{ \frac{\alpha}{2}, \frac{\epsilon}{2} \}$ such that, for any $X(0)$ satisfying $\maxnorm{ X(0)-X^*}<\zeta$, $\maxnorm{X(t)\!-\!X^*}\!<\!\epsilon$ for any $t\!\ge\! 0$, i.e., $X^*$ is locally stable.

Now we prove (ii)(a)$\,\Rightarrow$\,(ii)(b). We first establish the convergence of the solution $X(t)$ to some set of structurally balanced states via the LaSalle invariance principle. For simplicity, denote $X^{+}=\fHbM (X)$. The map $\fHbM(X)$ is continuous for any $X\in \Sssymm $ and, by Proposition~\ref{prop:HbM-finite-time-behav}, for any given $X(0)\in \Sssymm $, $\maxnorm{X(t)}\le \maxnorm{X(0)}$ for any $t\in \mathbb{Z}_{\ge 0}$. In addition, letting $\delta = \liminf\limits_{t\to\infty} \min\limits_{i,j}|X_{ij}(t)|>0$, we see that there exists $\tilde{t}\in \mathbb{Z}_{\ge 0}$ such that $\min\limits_{i,j}|X_{ij}(t)|\ge \delta/2$ for any $t\ge \tilde{t}$.
Therefore, the set 
\begin{align*}
G_c =  \Big{\{}  X\in \Sssymm \,\Big|\, & \min_{i,j}|X_{ij}|\ge \delta/2,\\
                & \maxnorm{X}\le \maxnorm{X(0)}\Big{\}}
\end{align*}
is a compact subset of $\Sssymm $ and $X(t)\in G_c$ for any $t\ge
\tilde{t}$. Thirdly, define $V(X)=\maxnorm{X}$. The
function $V$ is continuous on $\Sssymm $ and, by
Proposition~\ref{prop:HbM-finite-time-behav}, satisfies $V( X^+
)-V(X)\le 0$ for any $X\in \Sssymm $. According to the extended
LaSalle invariance principle  in Theorem~2
of~\cite{WM-FB:17s}, $X(t)$ converges to the largest invariant set $M$
of the set $E=\setdef{X\in G_c}{V( X^+ )-V(X)=0}$.

Now we characterize the largest invariant set $M$. For any $X\in M\subset E$, $V(X^{+})=V(X)=\maxnorm{X}$. Suppose $\lvert X^{+}_{ij}\rvert=\max\limits_{k,\ell}\lvert X^{+}_{k\ell}\rvert$. Since $X^+=\fHbM(X)$, we have
\vspace{-0.2cm}
\begin{equation}
\label{ineq11}
\begin{split}
\lvert X^{+}_{ij}\rvert 
&\leq \frac{1}{\norm{X_{i*}}_{1}}\sum\limits_{\ell=1}^{n}\lvert X_{i\ell}\rvert\lvert X_{j\ell}\rvert\le \maxnorm{X}.
\end{split}    
\end{equation}
In order for all these inequalities to hold with equality and noticing that $|X_{i\ell}|>0$ for any $\ell$ since $X\in G_c$, $X$ must satisfy that
\begin{enumerate}[label=(\alph*)]
   \item $X_{i*}$ and $X_{j*}$ have the same or opposite sign pattern, i.e., $\sign{(X_{i*})}=\pm\sign{(X_{j*})}$, \label{condM1}
   \item All entries of $X_{j*}$ have the magnitude $\maxnorm{X}$. \label{condM2}
\end{enumerate}
Therefore, for any $X\in E$, there exist some $i$ and $j$ such that the
aforementioned conditions~(a) and~(b) hold. Moreover, since the set $M$ is
invariant, $X\in M$ implies $X^+\in M\subset E$. Applying
Condition~\ref{condM2} to $X^+$, there exists a $\tilde{j}$ such that, for
any $p$, $\lvert X^{+}_{\tilde{j}p}\rvert=\maxnorm{X^{+}}=\maxnorm{X}$. In
order for $\lvert X^{+}_{\tilde{j}p}\rvert=\maxnorm{X}$ to hold, following
the same argument on the conditions such that the
inequalities~\eqref{ineq11} become equalities, we know that, for any $p$,
$\sign{(X_{\tilde{j}*})}=\pm\sign{(X_{p*})}$ and $\lvert
X_{pk}\rvert=\maxnorm{X}$ for any $k$. As these relationships hold for any
$p$, we conclude that for any $i,j \in\until{n}$, $X_{i*}$ and $X_{j*}$
must have the same or the opposite sign pattern. Let $\alpha=\maxnorm{X}$
and $b=\sign(X_{1*}^{\top})$. Each row of $X$ is thereby equal to either
$\alpha b^{\top}$ or $-\alpha b^{\top}$. Therefore, $X$ is of the form
$X=\alpha cb^{\top}$, where $c\in \{-1,1\}^n$. Moreover, since all the
diagonal entries of $X$ are positive, the column vector $c$ satisfies
$c_ib_i=1$ for any $i$, which implies $c=b$. In short, we have proved that
$X\in M$ leads to $X=\alpha bb^{\top}$. In addition, by
Theorem~\ref{thm:eq-set-QA}, any matrix $X=\alpha bb^{\top}$, with
$\alpha>0$ and $b\in \{-1,1\}^n$, is a fixed point of $\fHbM$ and is thus
invariant. Therefore, we conclude the compactness of
\begin{equation*}
M = \Big{\{} X=\alpha bb^{\top}\,\Big|\, \frac{\delta}{2} \le \alpha \le \maxnorm{X(0)}, b\in \{-1,1\}^n \Big{\}}.
\end{equation*}

For any $\hat{X}\in M$, since $\hat{X}$ satisfies social balance (see Theorem~\ref{thm:eq-set-QA}) and $\min_{i,j}|\hat{X}_{ij}| \ge \delta/2 >0$, there exists an open neighbor set defined as $\mathcal{U}(\hat{X})=\setdef{X = \hat{X}+\Delta}{\maxnorm{\Delta}<\min\limits_{i,j}|\hat{X}_{ij}|}$ such that any $X\in \mathcal{U}(\hat{X})$ satisfies social balance. According to Heine-Borel theorem, there exists a finite set $\{\hat{X}_1,\dots,\hat{X}_K\}\subset M$ such that $M\subset \cup_{k=1}^K \mathcal{U}(\hat{X}_k)$. Since $\cup_{k=1}^K \mathcal{U}(\hat{X}_k)$ is an open set, there exists $\epsilon>0$ such that the neighbor set of $M$, defined as $\mathcal{U}(M,\epsilon)=\setdef{X\in \Sssymm }{\maxnorm{X-M}<\epsilon}$,
satisfies that $\mathcal{U}(M,\epsilon)\subset \cup_{k=1}^K \mathcal{U}(\hat{X}_k)$ and thereby any $X\in \mathcal{U}(M,\epsilon)$ satisfies social balance.

Since $X(t)\to M$ as $t\to \infty$, there exists $t_0\in \mathbb{Z}_{\ge 0}$ such that $X(t)\in \mathcal{U}(M,\epsilon)$ for any $t\ge t_0$. Therefore, $X(t)$ satisfies social balance for any $t\ge t_0$, which concludes the proof for (ii)(a)$\,\Rightarrow\,$(ii)(b).

Now we prove (ii)(b)$\,\Rightarrow\,$(ii)(c). Suppose $G\big(X(t_0)\big)$ satisfies social balance for some $t_0>0$. If $\maxnorm{X(t_0)}=\minnorm{X(t_0)}$, then there exists some $\alpha>0$ such that $X(t_0)=\alpha B$, where $B\in \{-1,1\}^{n\times n}$. Since $G\big( X(t_0) \big)$ satisfies social balance, we have $B_{ii}>0$ and $B_{j*}=\pm B_{1*}$, which in turn implies that $B=B_{1*}^{\top}B_{1*}=\sign(X_{1*}(t_0))^{\top}\sign(X_{1*}(t_0))$. Therefore, $X(t_0)$ is already a rank-one fixed point in the set $\QHbM$.

Suppose $G\big(X(t_0)\big)$ satisfies social balance but $\maxnorm{X(t_0)}>\minnorm{X(t_0)}$. For any $t\ge t_0$, let $|X_{pq}(t)|=\minnorm{X(t)}$. We have that, for any $i$ and $j$,
\begin{align*}
 |X_{jp}& (t+1)|= \frac{1}{\lVert X_{j*}(t) \rVert_1}\sum_{k=1}^n |X_{jk}(t)||X_{pk}(t)|\\
        & \le \frac{|X_{jq}(t)|}{\lVert X_{j*}(t)\rVert_1}|X_{pq}(t)|+\Big( 1\!-\! \frac{|X_{jq}(t)|}{\lVert X_{j*}(t) \rVert_1} \Big)\maxnorm{X(t)}\\
        & \le \maxnorm{X(t)} - \frac{|X_{jq}(t)|}{\lVert X_{j*}(t) \rVert_1}\big( \maxnorm{X(t)}-\minnorm{X(t)} \big)\\
        & \le \maxnorm{X(t)}\! -\!\frac{\minnorm{X(t)}}{n\maxnorm{X(t)}}\big( \maxnorm{X(t)}\!-\!\minnorm{X(t)} \big),
\end{align*}
and, similarly,
\begin{align*}
 &|X_{ij}  (t+2)| = \frac{1}{\lVert X_{i*}(t\!+\! 1) \rVert_1}\sum_{k=1}^n |X_{ik}(t\!+\!1)||X_{jk}(t\!+\! 1)|\\
        & \le \frac{|X_{ip}(t + 1)|}{\lVert X_{i*}(t + 1) \rVert_1} |X_{jp}(t + 1)| \\
        &\quad + \Big( 1-\frac{|X_{ip}(t + 1)|}{\lVert X_{i*}(t + 1) \rVert_1} \Big)\maxnorm{X(t + 1)} \\
        &  \le \frac{|X_{ip}(t\! +\! 1)|}{\lVert X_{i*}(t\! +\! 1) \rVert_1} |X_{jp}(t\!+\!1)| 
  \!+\! \Big(\!1\!-\!\frac{|X_{ip}(t\! +\! 1)|}{\lVert X_{i*}(t\!+\!1) \rVert_1}\! \Big)\maxnorm{X(t)} \\
        & = \maxnorm{X(t)}  - \frac{|X_{ip}(t\! +\! 1)|}{\lVert X_{i*}(t\! +\! 1) \rVert_1}\big( \maxnorm{X(t)}-|X_{jp}(t\! +\! 1)| \big)\\
        &\le \maxnorm{X(t)} \\
        & \quad -\! \frac{\minnorm{X(t+1)}}{n\maxnorm{X(t+1)}}\frac{\minnorm{X(t)}}{n\maxnorm{X(t)}}\left( \maxnorm{X(t)}\!-\!\minnorm{X(t)} \right)\\
        & \le \maxnorm{X(t)} \!-\! \frac{\minnorm{X(t)}^2}{n^2\maxnorm{X(t)}^2}\big( \maxnorm{X(t)}\!-\!\minnorm{X(t)} \!\big).
\end{align*}
Therefore, 
\begin{align*}
& \maxnorm{X(t+2)}-\minnorm{X(t+2)} \\
&\qquad \le \Big( 1-\frac{\minnorm{X(t)}^2}{n^2 \maxnorm{X(t)}^2} \Big)\big( \maxnorm{X(t)} - \minnorm{X(t)} \big) \\
&\qquad \le \Big( 1-\frac{\minnorm{X(t_0)}^2}{n^2 \maxnorm{X(t_0)}^2} \Big)\big( \maxnorm{X(t)} - \minnorm{X(t)} \big).
\end{align*}
Now we have established the exponential convergence of $\maxnorm{X(t)}-\minnorm{X(t)}$ to $0$. Therefore, there exists $\alpha>0$ such that $\lim_{t\to \infty}|X_{ij}(t)|=\alpha$ for any $i,j$. Moreover, since $\sign(X(t))=\sign(X(t_0))$ for any $t\ge t_0$, we have $\lim_{t\to\infty}X(t)=\alpha bb^{\top}$, where $b=\sign(X_{1*}(t_0))^{\top}$. This concludes the proof for (ii)(b)$\,\Rightarrow\,$(ii)(c).

The proof for (ii)(b)$\,\Rightarrow\,$(ii)(a) is straightforward. If $G\big( X(t_0) \big)$ satisfies social balance, then, according to claim C.2), $\minnorm{X(t)}\ge \minnorm{X(t_0)}$ for any $t\ge t_0$, which means that $\liminf\limits_{t\to \infty}\min_{ij}|X_{ij}(t)|\ge \minnorm{X(t_0)}>0$.

Now we prove (ii)(c)$\,\Rightarrow\,$(ii)(b). Suppose $X(t)\to X^*$ as $t\to \infty$. For any $X^*\in \QHbM$ of rank one, there exists $\alpha>0$ and $b\in \{-1,1\}^n$ such that $X^* = \alpha bb^{\top}$. Since $\alpha>0$, there exists a neighbor set $\mathcal{U}(X^*)$ such that for any $X\in \mathcal{U}(X^*)$, $\sign{X}=\sign{X^*}$, which implies that, for any $X\in \mathcal{U}(X^*)$, $G(X)$ satisfies social balance. Moreover, since $X(t)\to X^*$, there exists $t_0>0$ such that $X(t)\in \mathcal{U}(X^*)$ for any $t\ge t_0$. Therefore, $G\big(X(t)\big)$ achieves social balance at $t_0$. This concludes the proof. \qed

As Theorem~\ref{theoremMain} points out, the appraisal matrix $X(t)$ converge to some rank-one matrix $\alpha bb^{\top}$ if and only if $X(t)$ achieves social balance (see Defintiion~\ref{defSB}) at some time $t_0$. The mathematical intuition behind the convergence to rank-one matrices is that, after achieving social balance, the quantity $\max_{ij}|X_{ij}(t)|-\min_{k\ell}|X_{k\ell}(t)|$ is monotonically vanishing. In reality, various factors such as noisy disturbances and individual prejudice (see~\cite{NEF-ECJ:90}) may prevent the appraisal matrix from converging to rank-one matrices.

Monte-Carlo validation of the non-vanishing appraisal condition indicates that statement~(ii)(b) of Theorem~\ref{theoremMain} holds for generic initial conditions. The detailed simulation results are presented in Section~\ref{section:discussion-simulation}. In fact, there exist some counter examples of $X(0)$ with which the non-vanishing condition on the solution $X(t)$ does not hold. \emph{Example~1}: if $X(0)$ is block-diagonal, then the dynamics of the blocks are decoupled. While statement~(ii) of Theorem~\ref{theoremMain} still holds block-wisely, the non-vanishing condition on the entire matrix $X(t)$ does not hold; \emph{Example~2}: if all the off-diagonal entries of $X(0)\in \mathbb{R}^{n\times n}$ are equal to some $-b<0$ and all the diagonal entries are equal to $a=(n-2)b/2$, one can check by computation that $X(1)$ becomes a diagonal matrix with strictly positive diagonals, i.e., $X(1)$ is a rank-$n$ fixed point and therefore the non-vanishing condition does not hold. However, for both Example~1 and~2, the sets of initial conditions are zero-measure and simulation results indicate that the zero-pattern of $X(t)$ with those specifically constructed $X(0)$ are not robust under perturbation: For Example~1, if $X(0)$ has two diagonal blocks, any perturbation of any of its zero-entries render the convergence of $X(t)$ to a rank-one matrix, and therefore the non-vanishing appraisal condition holds again; For Example~2, under any perturbation of any entry of $X(0)$, $X(t)$ converges to a rank-one matrix and the non-vanishing appraisal condition holds as well. Moreover, even for Example~1 and~2, the systems are still well-behaved and the solutions $X(t)$ achieve social balance with $k$ isolated subgraphs, as defined in Remark~\ref{remark:k-block-balance}. 

We end this section with some remarks on the homophily-based model.

\begin{remark}[Sufficient conditions for non-vanishing appraisals]
Since the non-vanishing appraisal condition is satisfied if $X(t)$ achieves social balance at finite time, by writing down the closed-form expressions of $X(1)$ and $X(2)$ and applying Lemma~\ref{numCond2}, we obtain the following sufficient conditions on the initial appraisals $X$ for non-vanishing appraisals: (i) either $(X_{i*}X_{1*}^{\top})(X_{1*}X_{j*}^{\top})(X_{i*}X_{j*}^{\top})>0$ for any $i,j$, (ii) or $(X_{i*}X^{\top}XX_{1*}^{\top})(X_{1*}X^{\top}XX_{j*}^{\top})(X_{i*}X^{\top}XX_{j*}^{\top})>0$ for any $i,j$. Here the condition~(i) ((ii) resp.) corresponds to the case when $X(1)$ ($X(2)$ resp.) is structurally balanced. For both condition~(i) and~(ii), the set of initial appraisal matrices $X$ have non-zero measure.
\end{remark}

\smallskip
\begin{remark}
Our homophily-based model exhibits the following somehow unrealistic behavior: for any $X(0)\in \Snz$, the solution $X(t)$ immediately becomes sign-symmetric at time step $1$. However, if we adopt a simple modification by considering individual memory, i.e., if the dynamics are given by
\begin{align}\label{eq:HbM-memory}
X(t+1) = \epsilon\, \fHbM(X(t)) + (1-\epsilon) X(t),
\end{align}
for some $\epsilon\in (0,1]$, then, following the same argument as in the proofs for Proposition~\ref{prop:HbM-finite-time-behav}, Theorem~\ref{thm:eq-set-QA}, and Theorem~\ref{theoremMain}, we conclude that
\begin{enumerate}
\item The set $\mathcal{S}_{\text{pos-diag}}=\{ X\in \mathbb{R}^{n\times n}\,|\, X_{ii}>0 \text{ for any }i\}$ is invariant under dynamics~\eqref{eq:HbM-memory};
\item Theorem~\ref{thm:eq-set-QA} still holds, while statements~(ii)-(iv) of Proposition~\ref{prop:HbM-finite-time-behav} and Theorem~\ref{theoremMain} still hold for any $X(0)\in \mathcal{S}_{\text{pos-diag}}$.
\end{enumerate} 
\end{remark}
The proof is provided in the technical report~\cite{WM-PCV-NEF-FB:17f-arxiv}.

\section{Influence-based Model}
\label{sectionModelB}
In this section, we propose the \emph{influence-based model} (IbM) and present some important theoretical results parallel to the results on the homophily-based model.
\smallskip

\begin{definition}[Influence-based model]
  Given an initial appraisal matrix
  $X(0)\in\Srsymm\subset\real^{n\times{n}}$, the influence-based model
  is defined by:
  \begin{equation} 
    \label{eq:IbM}
    X(t+1)=\diag(\lvert X(t)\rvert \vect{1}_{n})^{-1}X(t)X(t).
  \end{equation}
\end{definition}

\begin{remark}[Interpretation]
Compared with the homophily-based model~\eqref{eq:HbM}, the only difference here is that the term $X(t)X(t)^{\top}$ on the right-hand side of~\eqref{eq:HbM} is changed to $X(t)X(t)$. Equation~\eqref{eq:IbM} now describes an interpersonal influence process: Individuals adjust their appraisals of each other via the opinion dynamics $X(t+1)=W(t)X(t)$. Here the opinion of each individual is how she/he appraise every one in the group, and each $W_{ij}(t)$ denotes the weight that individual $i$ assigns to individual $j$'s opinions. The construction of the influence matrix $W(t)=\diag\big( |X(t)|\vect{1}_n \big)^{-1}X(t)$ implies that the interpersonal influences are proportional to the interpersonal appraisals. 
\end{remark}

Next, we present some results on the invariant set and finite-time behavior of the influence-based model.

\begin{proposition}[Finite-time Properties of the IbM]
  \label{prop:IbM-finite-time-behav}
Consider the dynamical system~\eqref{eq:IbM} and define $\fIbM (X)= \diag(\lvert X\rvert \vect{1}_{n})^{-1}XX$.
Pick any $X_0\in \Srsymm$. The following statements hold: 
\begin{enumerate}
\item \label{fact:IbM-inv-set} the map $\fIbM$ is well-defined for any $X\in \Snz$ and maps $\Srsymm$ to $\Srsymm$;
\item \label{fact:IbM-well-defined} the solution $X(t)$, $t\in \mathbb{Z}_{\ge 0}$, to equation~\eqref{eq:IbM} from initial condition $X(0)=X_0$ exists and is unique;
\item \label{fact:IbM-upper-bound} the max norm of $X(t)$ satisfies 
\begin{equation*}
\maxnorm{X(t+1)} \le \maxnorm{X(t)} \le \maxnorm{X(0)};
\end{equation*}
\item \label{fact:IbM-inv-scaling}for any $c > 0$, the trajectory $cX(t)$ is the solution to equation~\eqref{eq:IbM} from initial condition $X(0)=cX_0$.
\end{enumerate}
\end{proposition}

\begin{proof}
Denote $X^+=\fIbM(X)$ for simplicity. Following the same argument as in the proof of Proposition~\ref{prop:HbM-finite-time-behav}, we know that $\fIbM$ is well-defined for any $X\in \Snz$. For any $X\in \Srsymm$, there exists $\gamma \succ \vect{0}_n$ such that $\diag(\gamma)X=X^{\top}\diag(\gamma)$. Therefore,
\begin{align*}
X_{ii}^+ & \!=\! \frac{1}{\lVert X_{i*} \rVert_1}\! \sum_k X_{ik}X_{ki}\!=\!\frac{1}{\lVert X_{i*} \rVert_1}\! \sum_{k}\! \frac{\gamma_i}{\gamma_k}X_{ik}^2\!>\!0,\text{ and }\\
X_{ij}^+ & \!=\! \frac{1}{\lVert X_{i*} \rVert_1} \frac{\gamma_j}{\gamma_i}\!\sum_k\! X_{jk}X_{ki}\!=\!\frac{\lVert X_{j*} \rVert_1 \gamma_j}{\lVert X_{i*} \rVert_1 \gamma_i}X_{ji}^+.
\end{align*} 
Let $\tilde{\gamma}=\diag\big( |X|\vect{1}_n \big)\gamma$, then we have $\diag(\tilde{\gamma})X^+={X^{+}}^{\top}\diag(\tilde{\gamma})$. Therefore, $X^+=\fIbM(X)\in \Srsymm$. This concludes the proof of statement~\ref{fact:IbM-inv-set}. Statements~\ref{fact:IbM-well-defined} is a direct consequence of statement~\ref{fact:IbM-inv-set}. Moreover,
\begin{align*}
\big| X_{ij}^+ \big| & = \frac{1}{\lVert X_{i*} \rVert_1} \Big| \sum_k X_{ik}X_{kj} \Big| \le \frac{1}{\lVert X_{i*} \rVert_1} \sum_k |X_{ik}||X_{kj}|\\
                     & \le \max_k |X_{kj}| \le \maxnorm{X} 
\end{align*} 
immediately lead to statement~\ref{fact:IbM-upper-bound}. Statement~\ref{fact:IbM-inv-scaling} is a straightforward observation obtained from equation~\eqref{eq:IbM}.
\end{proof}

Notice that $\Snz$ is not an invariant set of the map $\fIbM$. For example,
\begin{equation*}
X(0) =  
\begin{bmatrix}
1 & 2\\
-0.5 & -1
\end{bmatrix}\in\Snz
\end{equation*} 
leads to $X(1)\notin\Snz$ and, moreover, $\fIbM(X(1))$ is not defined. 
For the influence-based model, we consider $\Srsymm$ as the domain of system~\eqref{eq:IbM} due to its invariance under the map $\fIbM$. According to Proposition~\ref{prop:IbM-finite-time-behav}, for any $X(0) \in \Srsymm$, each entry of $|X(t)|$ is uniformly upper bounded, which is a desired property the previous models in~\cite{SAM-JK-RDK-SHS:11,VAT-PVD-PDL:13} do not have.

The following theorem characterizes the set of fixed points of the map $\fIbM$ in $\Srsymm$.

\begin{theorem}[Fixed points and social balance]
\label{thm:eq-set-QB}
Consider system~\eqref{eq:IbM} in domain $\Srsymm$. Define
\begin{equation*}
\begin{split}
      &\QIbM \\
      &\text{  }= \Big{\{} PYP^\top  \in \Srsymm \,\Big|\, P \emph{ is a permutation matrix,}\\
      &\qquad \text{ }Y\emph{ is a block diagonal matrix with blocks of the}\\
      &\qquad \text{ }\emph{form }\sign(w)w^\top,\; w \in \real^{m}\text{ and }|w|\succ \vect{0}_m, m \leq n\Big{\}}.
\end{split}
\end{equation*}
Then the following statements hold:
\begin{enumerate}
	\item \label{prop:eqSetQB-AllFixedPoints} $\QIbM$ is the set of all the fixed points of system~\eqref{eq:IbM} in domain $\Srsymm$,
	\item \label{prop:eqSetQB-StructuralBalance} for any $X\in\QIbM$, $G(X)$ is composed by isolated complete subgraphs that satisfy social balance.
\end{enumerate}
\end{theorem}

\begin{proof}
We first prove that any $X^*\in\QIbM$ is a fixed point of system~\eqref{eq:IbM}. For any $w \in \real^{n}$ such that $|w|\succ \vect{0}_n$, the matrix $Y=\sign(w)w^{\top}$ satisfies
\begin{align*}
& \fIbM(Y) \\
       & \quad = \diag(|\sign(w)w^{\top}|\vect{1}_{n})^{-1}(\sign(w)w^{\top})(\sign(w)w^{\top})\\
       & \quad = \sign(w)w^{\top} = Y.
\end{align*} 
Therefore, $Y=\sign(w)w^{\top}$ is a fixed point of system~\eqref{eq:IbM}. A simple observation is that, if $Y$ is a block diagonal matrix and each block takes the form $\sign(w)w^{\top}$, then $Y$ is a fixed point of $\fIbM$.
Moreover, given any fixed point $Y$, for any permutation matrix $P\in \real^{n\times n}$, since
\begin{align*}
PYP^{\top} & = P\diag(|Y|\vect{1}_n)^{-1}YYP^{\top}\\
           & = \diag(|PYP^{\top}|\vect{1}_n)^{-1}(PYP^{\top})(PYP^{\top})\\
           & = \fIbM(PYP^{\top}),
\end{align*}
any $X^*\in \QIbM$ is a fixed point of $\fIbM$.

Since, for any $Y=sign(w)w^{\top}$ with $|w|\succ 0$, we have  $Y_{jk}Y_{k\ell}Y_{\ell j}=|w_j w_k w_\ell|>0$, following the same line of argument in the proof for Theorem 3.4(ii), we conclude that any $X^*\in \QIbM$ is associated with a graph $G(X^{*})$ composed by isolated complete subgraphs that satisfy social balance. This proves statement~\ref{prop:eqSetQB-StructuralBalance}. 

Now we prove by induction that $\QIbM$ is actually the set of all the fixed point of system~\eqref{eq:IbM} in $\Srsymm$. We adopt the notations $\theta_i$ and $X_{j*,\theta_i}$ in the same way as defined in the proof of Theorem~\ref{thm:eq-set-QA}, and, in addition, define $X_{*j,\theta_i}$ as the $j$-th column of $X$ with all the $k$-th entry such that $k\notin \theta_i$ removed. One can check that the trivial case of $n=1$ is true. Suppose statement~\ref{prop:eqSetQB-AllFixedPoints} holds for any system with dimension $\tilde{n}<n$.

For system~\eqref{eq:IbM} with dimension $n$, suppose $X\in \Srsymm$ is a fixed point of the system~\eqref{eq:IbM}, i.e., $X = \fIbM(X)$. For any given $j$,
\begin{align*}
|X_{ij}| & = \frac{1}{\lVert X_{i*} \rVert_1}\Big| \sum_k X_{ik}X_{kj} \Big| \le \frac{1}{\lVert X_{i*} \rVert_1} \sum_k |X_{ik}||X_{kj}| \\
         & \le \max_k |X_{kj}|, \quad \text{for any }i,
\end{align*}
and there exists some $i$ such that $|X_{ij}|=\max_k |X_{kj}|$. Now we discuss two cases that cover all the possible $X$'s.

Case 1: $|X_{jj}|=\max_k |X_{kj}|$ and $|X_{ij}|<\max_k |X_{kj}|$ for any $i\neq j$. Since $X\in \Srsymm$, $X$ is sign-symmetric, 
\begin{equation*}
|X_{jj}| = \frac{1}{\lVert X_{j*} \rVert_1} \sum_k |X_{jk}| |X_{kj}|=\max_k |X_{kj}|.
\end{equation*}
Due to the second equality in the equations above, $|X_{kj}|=\max_{\ell} |X_{\ell j}|$ for any $k\in \theta_j$. Therefore, in Case 1, $i\notin \theta_j$ for any $i\neq j$, which in turn implies that $X_{ji}=X_{ij}=0$ for any $i\neq j$. As the consequence, there exists a permutation matrix $P$ such that
\begin{equation*}
PXP^{\top} = 
\begin{bmatrix}
X_{jj} & \vect{0}_{n-1}^{\top} \\
\vect{0}_{n-1} & \tilde{X}
\end{bmatrix},
\end{equation*}
where $\tilde{X}$ is an $(n-1)\times (n-1)$ matrix. Following the same line of argument in the Case 1 of the proof of Theorem~\ref{thm:eq-set-QA}, we conclude that $X\in \QIbM$.

Case 2: there exists $i\neq j$ such that $|X_{ij}|=\max_k |X_{kj}|$. For such $i$, we have $j\in \theta_i$. In addition, the equality below
\begin{equation*}
|X_{ij}| = \frac{1}{\lVert X_{i*} \rVert_1}\Big| \sum_k X_{ik}X_{kj} \Big|=\max_k |X_{kj}| 
\end{equation*}  
leads to the following two results:
\begin{enumerate}[label={R.\arabic*)}]
\item $\sign( X_{i*,\theta_i}) = \pm \sign(X_{*j,\theta_i}^{\top}) $;
\item $|X_{kj}| = \max_{\ell} |X_{\ell j}|$ for any $k\in \theta_i$.
\end{enumerate}
Result R.2) and $j\in \theta_i$ lead to $|X_{jj}|=\max_{\ell} |X_{\ell j}|$. Therefore, for any $k\in \theta_j$, $|X_{kj}|=\max_{\ell} |X_{\ell j}|$. Moreover, since $X$ is sign-symmetric, for any $k\notin \theta_j$, $X_{jk}=X_{kj}=0$. 

For any $i\in \theta_j$, since 
\begin{align*}
|X_{ij}| = \frac{1}{\lVert X_{i*} \rVert_1} \Big| \sum_k X_{ik}X_{kj} \Big| = \max_{\ell} |X_{\ell j}|,
\end{align*}
we have $|X_{kj}|=\max_{\ell} |X_{\ell j}|$ for any $k\in \theta_i$. Since $\max_{\ell}|X_{\ell j}|>0$ for any $X\in \Srsymm$ and $X_{kj}=0$ for any $k\in \theta_j$, we have $\theta_i \subset \theta_j$.

For any given $i\in \theta_j$, since $X\in \Srsymm$, we know that $X_{ii}>0$ and $X_{ji}>0$. Apply the same argument for the $j$-th column in Case 2 to the $i$-th column, we conclude that $X_{ii} = \max_{\ell} |X_{\ell i}|$ and $|X_{ji}|=\max_{\ell} |X_{\ell i}|$, the latter of which in turn implies that $|X_{ki}|=\max_{\ell} |X_{\ell i}|$ for any $k\in \theta_j$. Moreover, since $X_{ii}=\max_{\ell} |X_{\ell i}|$ leads to $|X_{ki}|=\max_{\ell} |X_{\ell i}|$ for any $k\in \theta_i$ and $X_{ik}=X_{ki}=0$ for any $k\notin \theta_i$, we have $\theta_j \subset \theta_i$. Since we already get $\theta_i\subset \theta_j$, we conclude that $\theta_j = \theta_i$ for any $i\in \theta_j$. Now we have proved that graph $G(X)$ can be partitioned into two isolated subgraphs with the node sets $\theta_j$ and $\{1,\dots,n\}\setminus \theta_j$ respectively. In addition, due to Result R.1) and the facts that $\theta_i=\theta_j$ and $X$ is sign-symmetric, we obtain that $\sign(X_{i*},\theta_j)=\sign(X_{*i,\theta_j}^{\top})=\pm \sign(X_{*j,\theta_j}^{\top})$ for all $i\in \theta_j$.

Taking together all the results we have obtained for Case 2, we conclude that, for any given $j$ in Case 2: (1) $|X_{kj}|=\max_{\ell}|X_{\ell j}|$ for any $k\in \theta_j$ and $X_{kj}=X_{jk}=0$ for any $k\notin \theta_j$; (2) For any $i\in \theta_j$, $\theta_i=\theta_j$. In addition, $|X_{ki}|=\max_{\ell} |X_{\ell i}|$ for any $k\in \theta_j$ and $X_{ki}=X_{ik}=0$ for any $k\notin \theta_j$; (3) For any $i\in \theta_j$, $\sign(X_{*i})=\sign(X_{*j})$. Denote by $|\theta_j|$ the cardinality of $\theta_j$ and define the $|\theta_j|\times |\theta_j|$ matrix $X^{(\theta_j)} = \sign( w^{(\theta_j)} ){w^{(\theta_j)}}^{\top}$, where $w^{(\theta_j)}=X_{j*,\theta_j}^{\top}$. If $|\theta_j|=n$, then $X$ is already of the form $\sign(w)w^{\top}$ and thus we have $X\in \QIbM$. If $|\theta_j|\neq n$, there exists a permutation matrix $P$ such that 
\begin{equation*}
PXP^{\top} = 
\begin{bmatrix}
X^{(\theta_j)} & \vect{0}_{|\theta_j|\times (n-|\theta_j|)} \\
\vect{0}_{(n-|\theta_j|) \times |\theta_j|} & \tilde{X}
\end{bmatrix}.
\end{equation*}
Following the line of argument in Case 1, we have $X\in \QIbM$. This concludes the proof for statement~\ref{prop:eqSetQB-AllFixedPoints}.
\end{proof}
\smallskip

\begin{remark}
  The proof of Theorem~\ref{thm:eq-set-QB} implies that
    $\QIbM$ is actually the set of all the fixed points of the map $\fIbM$
    in $\Sssymm$. However, the set $\QIbM$ does not contain all the fixed
    points in $\Snz$. For example, let $X=\alpha bb^{\top}$ for some
    $\alpha>0$ and $b\in\{-1,+1\}^{n}$. Then, pick one $i\in\until{n}$ and
    set $X_{*i}= \vect{0}_n$. It can be easily verified that $X=\fIbM(X)$
    but $X\notin\QIbM$.
\end{remark}

Now we present the main results on the convergence of the appraisal network to social balance.
\smallskip

\begin{theorem}[Convergence and social balance in the IbM]
\label{theoremMain1}
Consider the influence-based model given by equation~\eqref{eq:IbM}. The following statements hold:
\begin{enumerate}
\item Each element in $\QIbM$ of rank one is a locally stable fixed point of $\fIbM$;
\item For any $X(0)\in \Srsymm$, the following three statements are equivalent:
   \begin{enumerate}
   \item the solution $X(t)$ satisfies the non-vanishing appraisal condition given by Definition~\ref{def:lower-bounded-appraisal-condition};
   \item there exists $t_0\ge 0$ such that $G(X(t))$ satisfies social balance for all $t\ge t_0$;
   \item there exists $X^*\in \QIbM$ of rank one such that $\lim_{t\to \infty} X(t) = X^*$.
   \end{enumerate}
\end{enumerate}
\end{theorem}

\begin{proof}
We start by proving the following two claims. For any given $t_0\ge 0$, if all the entries of $X(t_0)$ are non-zero and $G(X(t_0))$ satisfies social balance, then,
\begin{enumerate}[label={C.\arabic*)}]
\item\label{claim-modelB-main-theorem-1} for any $t\ge t_0$, $G(X(t))$ satisfies social balance and $\sign(X(t))=\sign(X(t_0))$;
\item\label{claim-modelB-main-theorem-2} for any $t\ge t_0$, $\maxnorm{X(t)}$ is non-increasing and $\minnorm{X(t)}$ is non-decreasing.
\end{enumerate}
Claim~\ref{claim-modelB-main-theorem-1} is proved in the same way as in the proof of Theorem~\ref{theoremMain}. Suppose $G(X(t_0))$ achieves social balance. For any $i$ and $j$, since
\begin{align*}
X_{ij}(t_0+1) = \frac{1}{\lVert X_{i*}(t_0) \rVert_1}\sum_k X_{ik}(t_0)X_{kj}(t_0),
\end{align*}
we have 
\begin{align*}
\sign(X_{ij}(t_0+1)) = \sign\left( \sum_k X_{ik}(t_0)X_{kj}(t_0) \right).
\end{align*}
Since $G(X(t_0))$ achieves social balance, for any $k$,
\begin{align*}
\sign(X_{ik}(t_0)X_{kj}(t_0)) = \sign(X_{ji}(t_0)) = \sign(X_{ij}(t_0)).
\end{align*} 
Therefore, $\sign(X(t_0+1))=\sign(X(t_0))$. This concludes the proof of claim~\ref{claim-modelB-main-theorem-1}.

For any $t\ge t_0$, since $G(X(t))$ satisfies social balance,
\begin{equation*}
|X_{ij}(t+1)| = \frac{1}{\norm{X_{i*}(t)}_{1}}\sum_{\ell=1}^n |X_{i\ell}(t)||X_{\ell j}(t)|,
\end{equation*}
for any $i$ and $j$. Therefore, for any given $j$, the previous expression leads to the following two inequalities:\\
$\min\limits_{\ell}|X_{\ell j}(t\!+\!1)|\! \ge\! \min\limits_{\ell} |X_{\ell j}(t)|$; $\max\limits_{\ell}|X_{\ell j}(t\!+\!1)|\! \le\! \max\limits_{\ell} |X_{\ell j}(t)|$. Let $\maxnorm{X(t)}=\max_{i,j}|X_{ij}(t)|$ and $\minnorm{X(t)}=\min_{i,j}|X_{ij}(t)|$. We have $\maxnorm{X(t+1)}\le \maxnorm{X(t)}$ and $\minnorm{X(t+1)}\ge \minnorm{X(t)}$ for any $t\ge t_0$. This concludes the proof of Claim~\ref{claim-modelB-main-theorem-2}.

Now we prove statement~(i), i.e., each $\hat{X}\in \QIbM$ with rank 1 is locally stable. Let $\hat{X}=\sign(w)w^{\top}$, where $|w|\succ \vect{0}_n$. For any matrix $\Delta\in\real^{n\times{n}}$ such that for any $k\in\until{n}$, $\delta_{k}=\max_{i}|\Delta_{ik}|<|w_{k}|$, we have $\sign(\hat{X}_{*k}+\Delta_{*k})=\sign(\hat{X}_{*k})$. Due to claim \ref{claim-modelB-main-theorem-1} and the proof of claim \ref{claim-modelB-main-theorem-2}, we know that, for $X(0)=\hat{X}+\Delta$, $X(t)$ satisfies that, for any $t\ge 0$,
\begin{enumerate}[label=(\arabic*)]
\item $\sign(X(t))=\sign(X(0))=\sign(\hat{X})$;
\item $|w_{k}|\!-\!\delta_{k}\!\le\! \min_{i}|X_{ik}(t)|\!\le\! \max_{i}|X_{ik}(t)|\!\le\! |w_{k}|\!+\!\delta_{k}$.
\end{enumerate}
Therefore, for any $i$, $X_{ik}(t)$ is of the form $\alpha_{ik}\!(t)\!\sign(\hat{X}_{ik})$, where $0<|w_{k}|-\delta_{k} \le \alpha_{ik}(t)\le |w_{k}|+\delta_{k}$. We have
\begin{align*}
\maxnorm{X(t)-\hat{X}} & \!=\! \max_{ij} \big| \alpha_{ij}(t)\sign(\hat{X}_{ij})\! -\! |w_{j}| \sign(\hat{X}_{ij}) \big|\\
& \!=\! \max_{ij} \big| \alpha_{ij}(t) - |w_{j}| \big|\!\le\! \delta,
\end{align*}
where $\delta=\max\limits_{k}\delta_{k}$. Therefore, for any $\epsilon>0$, there exists $\delta=\min\{ \frac{\max_{k}|w_{k}|}{2}, \frac{\epsilon}{2} \}$ such that, for any $X(0)$ satisfying $\maxnorm{X(0)-X^*}<\delta$, $\maxnorm{X(t)-X^*}<\epsilon$ for any $t\ge 0$. That is, $\hat{X}$ is locally stable.

Now we prove (ii)(a) $\Rightarrow$ (ii)(b). For simplicity, denote $X^{+}=\fIbM (X)$. Firstly, one can easily check that $\fIbM(X)$ is continuous for any $X\in \Srsymm$. Secondly, for any $X(0)\in \Srsymm$ and any given $k\in\until{n}$, according to the proof of Proposition~\ref{prop:IbM-finite-time-behav}, $\lVert X_{*k}(t) \rVert_{\infty}\le \lVert X_{*k}(0) \rVert_{\infty}$ for any $t\in \mathbb{Z}_{\ge 0}$. In addition, let $\delta = \liminf\limits_{t\to\infty} \min\limits_{i,j}|X_{ij}(t)| >0$, then there exists $\tilde{t}\in \mathbb{Z}_{\ge 0}$ such that $\min\limits_{i,j}|X_{ij}(t)|\ge \delta/2$ for any $t\ge \tilde{t}$. Therefore, the set 
\begin{align*}
G_c =  \Big{\{}X\in \Srsymm \,\Big|\, & \min_{i,j}|X_{ij}|\ge \delta/2,\text{ and, for any }k,\\
&\norm{X_{*k}}_{\infty} \le \norm{X_{*k}(0)}_{\infty} \Big{\}}\\
\end{align*}
is a compact subset of $\Srsymm $ and $X(t)\in G_c$ for any $t\ge \tilde{t}$. Thirdly, define $V(X_{*k})=\lVert X_{*k} \rVert_{\infty}$. The function $V$ is continuous on $\Srsymm $ and, according to the proof of Proposition~\ref{prop:IbM-finite-time-behav}, satisfies $V( X_{*k}^+ )-V(X_{*k})\le 0$ for any $X\in \Srsymm $. According to the extended LaSalle invariance principle presented in Theorem~2 of~\cite{WM-FB:17s}, we conclude that, given any $X(0)\in \Srsymm$ such that $\liminf\limits_{t\to\infty} \min\limits_{i,j}|X_{ij}(t)| = \delta$, $X(t)$ converges to the largest invariant set $M$ of the set $E=\setdef{X\in G_c}{V( X_{*k}^+ )-V(X_{*k})=0\text{ for any }k}$.

Now we characterize the largest invariant set $M$. For any $X\in M\subset E$ and $k\in\until{n}$, $V(X_{*k}^{+})=V(X_{*k})=\norm{X_{*k}}_{\infty}$. Suppose $\lvert X^{+}_{ik}\rvert=\max\limits_{\ell}\lvert X^{+}_{\ell k}\rvert$. Since
\begin{equation}
\label{ineq111}
\begin{split}
\lvert X^{+}_{ik}\rvert &= \frac{1}{\norm{X_{i*}}_{1}}\left|\sum\limits_{\ell=1}^{n}X_{i\ell}X_{\ell k}\right|\\
&\leq \frac{1}{\norm{X_{i*}}_{1}}\sum\limits_{\ell=1}^{n}\lvert X_{i\ell}\rvert\lvert X_{\ell k}\rvert\leq \max\limits_{\ell}\lvert X_{\ell k} \rvert,
\end{split}    
\end{equation}
we need all of these inequalities to hold with equality 
Since $X\in G_c \subset \Srsymm$ implies $\lvert X_{j\ell}\rvert > 0$, for any $j,\ell\in\until{n}$, $X$ must satisfy that
\begin{enumerate}[label=(\alph*)]
   \item $X_{i*}$ and $X_{*k}$ have the same or opposite sign pattern, i.e., $\sign{(X_{*k})}=\sign{(X_{k*})}=\pm\sign{(X_{i*})}$, \label{condM11}
   \item All entries of $X_{*k}$ have magnitude $\norm{X_{*k}}_{\infty}$. \label{condM2}
\end{enumerate}
Moreover, since the set $M$ is invariant, $X\in M$ implies $X^+\in M\subset E$, which in turn implies that, for any $p$, $\lvert X^{+}_{pk}\rvert=\norm{X^{+}_{*k}}_{\infty}=\norm{X_{*k}}_{\infty}$. Following the same argument on the conditions such that the inequalities~\eqref{ineq111} become strict equalities, we know that, for any $p$, $\sign{(X_{p*})}=\pm\sign{(X_{*k}^{\top})}$ and $\lvert X_{pk}\rvert=\norm{X_{*k}}_{\infty}$ for any $k$. Using these relationships, we conclude that for any $i$ and $j$, $X_{i*}$ and $X_{j*}$ must have the same or the opposite sign pattern, and that $\lvert X_{ij}\rvert=\norm{X_{*j}}_{\infty}$. Let $w = X_{1*}^{\top}$. Each row of $X$ is thereby equal to either $w^{\top}$ or $-w^{\top}$. Therefore, $X$ is of the form $X=cw^{\top}$, where $c\in \{-1,1\}^n$. Moreover, since all the diagonal entries of $X$ are positive, the column vector $c$ must satisfy $c_iw_i=1$ for any $i$, which implies $c=\sign(w)$. Therefore, $X=\sign(w)w^{\top}$. Thus, since any matrix $X$ of the form $\sign(w)w^{\top}$, with $|w|\succ \vect{0}_n$, is a fixed point of system~\eqref{eq:IbM}, we conclude that 
\begin{equation*}
\begin{split}
M =& \setdef{X=\sign(w)w^{\top}}{\delta/2 \le w_{i}\le \maxnorm{X(0)},\\
&w\in\real^{n}\setminus\{\vect{0}_n\},\text{ for any }i\in\until{n}},
\end{split}
\end{equation*}
which is a compact subset of $\Srsymm$. Following the same line of argument in the proof of Theorem~\ref{theoremMain}, we conclude that there exists $\epsilon>0$ such that any $X$ in the neighbor set $\mathcal{U}(M,\epsilon)$ satisfies social balance.

Since $X(t)\to M$ as $t\to \infty$, there exists $t_0\in \mathbb{Z}_{\ge 0}$ such that $X(t)\in \mathcal{U}(M,\epsilon)$ for any $t\ge t_0$. Therefore, $X(t)$ satisfies social balance for any $t\ge t_0$, which concludes the proof for (ii)(a) $\Rightarrow$ (ii)(b). 

Now we proceed to prove (ii)(b) $\Rightarrow$ (ii)(c). If $G(X(t_0))$ satisfies social balance for some $t_0$ and $\max_{\ell}|X_{\ell j}(t_0)|=\min_{\ell}|X_{\ell j}(t_0)|$ for any $j$, then $X(t_0)=\sign(X_{*1}(t_0))X_{*1}(t_0)^{\top}\in \QIbM$, which is already a fixed point of system~\eqref{eq:IbM}. Suppose $G(X(t_0))$ satisfies social balance at time $t_0$ but there exists $j$ such that $\max_{\ell}|X_{\ell j}(t_0)|>\min_{\ell}|X_{\ell j}(t_0)|$. For any $t\ge t_0$, since
\begin{align*}
|X_{ij}(t+1)| & \ge \frac{\norm{X_{i*}(t)}_1 - \minnorm{X(t_0)}}{\norm{X_{i*}(t)}_1} \min_{\ell} |X_{\ell j}(t)| \\
                    & \quad + \frac{\minnorm{X(t_0)}}{\norm{X_{i*}(t)}_1} \max_{\ell} |X_{\ell j}(t)|,
\end{align*}
for any $i$, and $|X_{ij}(t+1)|\le \max_{\ell} |X_{\ell j}(t)|$, we have that
\begin{align*}
& \max_{\ell}  |X_{\ell j}(t+1)| - \min_{\ell} |X_{\ell j}(t+1)|\\
                   & \text{ }\le \Big( 1-\frac{\minnorm{X(t_0)}}{\norm{X_{i*}(t)}_1} \Big)\Big(\max_{\ell}  |X_{\ell j}(t)| - \min_{\ell} |X_{\ell j}(t)|\Big)\\
                   & \text{ }\le \Big( 1-\frac{\minnorm{X(t_0)}}{n\maxnorm{X(t_0)}} \Big)\Big(\max_{\ell}  |X_{\ell j}(t)| - \min_{\ell} |X_{\ell j}(t)|\Big).
\end{align*}
Therefore, for any given $j$,  after $t_0$, $\max_{\ell}  |X_{\ell j}(t)| - \min_{\ell} |X_{\ell j}(t)|$ exponentially converges to $0$. In addition, since $\sign(X(t))=\sign(X(t_0))$ for all $t\ge t_0$, we conclude that $X(t)$ converges to a matrix in the form $\sign(w)w^{\top}$. This concludes the proof for (ii)(b) $\Rightarrow$ (ii)(c).

(ii)(b) $\Rightarrow$ (ii)(a) and (ii)(c) $\Rightarrow$ (ii)(b) are proved following the same arguments as in the proof of Theorem~\ref{theoremMain}.
\end{proof}

\section{Further discussion and numerical simulations}\label{section:discussion-simulation}
\label{moredicussions}
\subsection{Numerical validation of the non-vanishing appraisal condition and model comparisons}
\label{sub_conv}

Monte Carlo validation indicates that, for the homophily-based model, the non-vanishing appraisal condition, given by Definition~\ref{def:lower-bounded-appraisal-condition}, holds for generic initial conditions in $\Snz$. By generic initial condition, we mean each of $X(0)$'s entries is independently randomly generated from the uniform distribution on some support $[-a,a]$. Since the homophily-based model is independent of scaling, we only need to consider the support $[-1,1]$.
For any randomly generated $X(0)\in\Snz \cap [-1,1]^{n\times n}$, define the random variable $Z~:~\Snz~\to~\{0,1\}$ as
\begin{align*}
Z(X(0))=
\begin{cases}
\displaystyle 1,\quad &\text{if } \min\limits_{100\le t\le 1000} \min\limits_{i,j} |X_{ij}(t)| \ge 0.001, \\
\displaystyle 0,\quad &\text{otherwise}.
\end{cases}
\end{align*}
Let $p=\mathbb{P}[Z(X(0))=1]$. For $N$ such independent random samples $Z_{1},\dots,Z_{N}$, define $\hat{p}_{N}=\sum^{N}_{i=1}Z_{i}/N$. For any accuracy $1-\varepsilon\in(0,1)$ and confidence level $1-\xi\in(0,1)$, $|\hat{p}_{N}-p|<\varepsilon$ with probability greater than $1-\xi$ if the Chernoff bound is satisfied: $N\geq\frac{1}{2\varepsilon^2}\log{\frac{2}{\xi}}$.
For $\varepsilon=\xi=0.01$, the bound is satisfied by $N=27000$. We ran the 27000 independent simulations of the homophily-based model with $n=8$, and found that $\hat{p}=1$. Therefore, we conclude that, for any generic initial condition $X(0)\in\Snz$, with 99$\%$ confidence level, there is at least 0.99 probability that every entry of $|X(t)|$ is lower bounded by a positive scalar (set to be 0.001 in this simulation) for all $t\in \{100,\dots,10000\}$.

We remark that the continuous-time homophily-based model~\cite{VAT-PVD-PDL:13} has a similar property that the interpersonal appraisals reach social balance in finite time, however they diverge later also in finite time.

The same Monte Carlo validation is also applied to the influence-based model, except that now the generic initial conditions $X(0)\in \Srsymm \subset \real^{n\times n}$ is generated by the following steps: 1) Randomly and independently generate the diagonal and the upper triangular entries of a matrix $\hat{X}\in \real^{n\times n}$ from the uniform distribution on $[-1,1]$; 2) Let $\hat{X}_{ij}=\hat{X}_{ji}$ for any $i>j$; 3) Randomly and independently generate the entries of a $n\times 1$ vector $\gamma$ from the uniform distribution on $[0,1]$; 4) Let $X(0)=\diag(\gamma)\hat{X}$.
We obtained that, for any initial condition $X(0)\in\Srsymm$, with 99$\%$ confidence level, there is at least 0.99 probability that every entry of $|X(t)|$ is uniformly strictly lower bounded from $0$ for all $t\in \{100,\dots,10000\}$. 

In the continuous-time influence-like model~\cite{SAM-JK-RDK-SHS:11,VAT-PVD-PDL:13}, when the initial appraisal matrix $X(0)$ is a normal matrix, i.e., when $X(0)X(0)^{\top}=X(0)^{\top}X(0)$, the appraisal network $G(X(t))$ almost surely reaches social balance only in the limit case when the network size $n$ tends to infinity. Compared with these models, besides the desired convergence property, our influenced-based model has the following advantages: 1) Unlike the set of normal matrices, of which the sociological meaning is not explicit, the almost-sure convergence to social balance in our influence-based model holds for any $X(0)=\diag(\gamma)\hat{X}$, where $\hat{X}$ is symmetric and $\diag(\gamma)$ has positive diagonals.  With the term $\diag(\gamma)$, our model allows for individuals' heterogeneous scaling of appraisals, which is sociologically more reasonable; 2) In our influence-based model, the almost-sure finite-time achievement of social balance holds for any finite network size $n$.

For both homophily-based and influence-based models, Monte Carlo validations with uniform but asymmetric initial appraisal distributions leads to the same results, but are not presented here due to the limit of space.

We further numerically estimate, for our influence-based model, the probability that the non-vanishing appraisal condition holds for generic initial conditions $X(0)\in \Snz\cap [-1,1]^{n\times n}$. According to Theorem~\ref{theoremMain1}, this probability is also the probability that the appraisal network converges to social balance. As shown in Fig.~\ref{fig:balance_prob_compare}, for the influence-based model, the probability of converging to social balance is quite low and decays to zero as the network size increases. Such feature indicates that, if system~\eqref{eq:HbM} and~\eqref{eq:IbM} correctly characterize the homophily and influence mechanisms respectively, then the homophily mechanism is a more universal explanation for the convergence of appraisal networks to social balance. That is, it is more probable that the empirically observed structurally balanced social networks are formed via the homophily mechanism rather than the influence mechanism.

\begin{figure}\label{fig:balance_prob_compare}
\begin{center}
\includegraphics[width=0.75\linewidth]{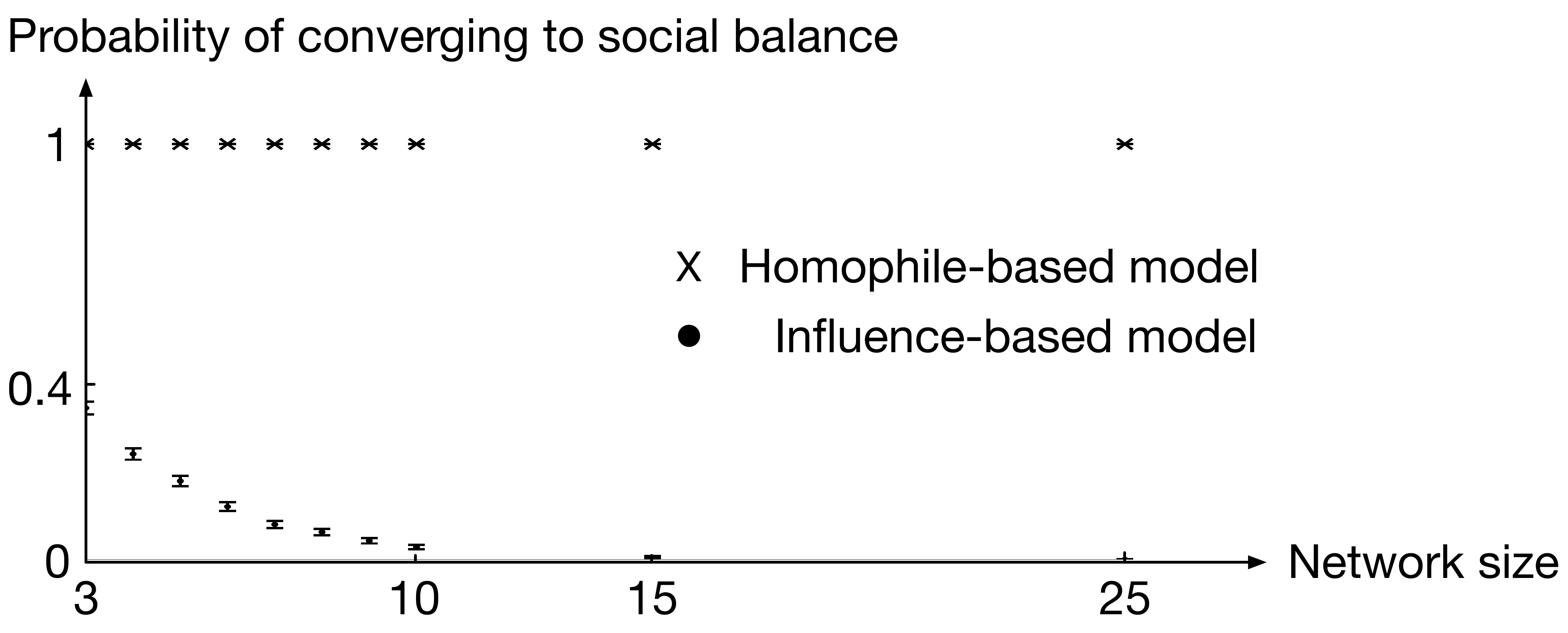}
\end{center}
\caption{Error-bar plot of the estimated probability of
    converging to social balance for both the homophily-based model and the
    influence model. For each network size, we run 1000 realizations, each
    with an initial condition $X(0)$ randomly generated from $\Snz\cap
    [-1.1]^{n\times n}$ in the same way as in the first paragraph of
    Section 5.1. Numerical convergence is determined by whether the
    non-vanishing appraisal condition holds. The error bars are taken as
    the estimated standard deviations of the probability estimation and
    turn our to be very small (0 for the homophily-based model).}
\end{figure}

\subsection{Social balance under perturbation}

\begin{figure} 
  \captionsetup[subfigure]{labelformat=empty}
  \centering
  \subfloat[no link added]{%
    \includegraphics[width=0.22\linewidth]{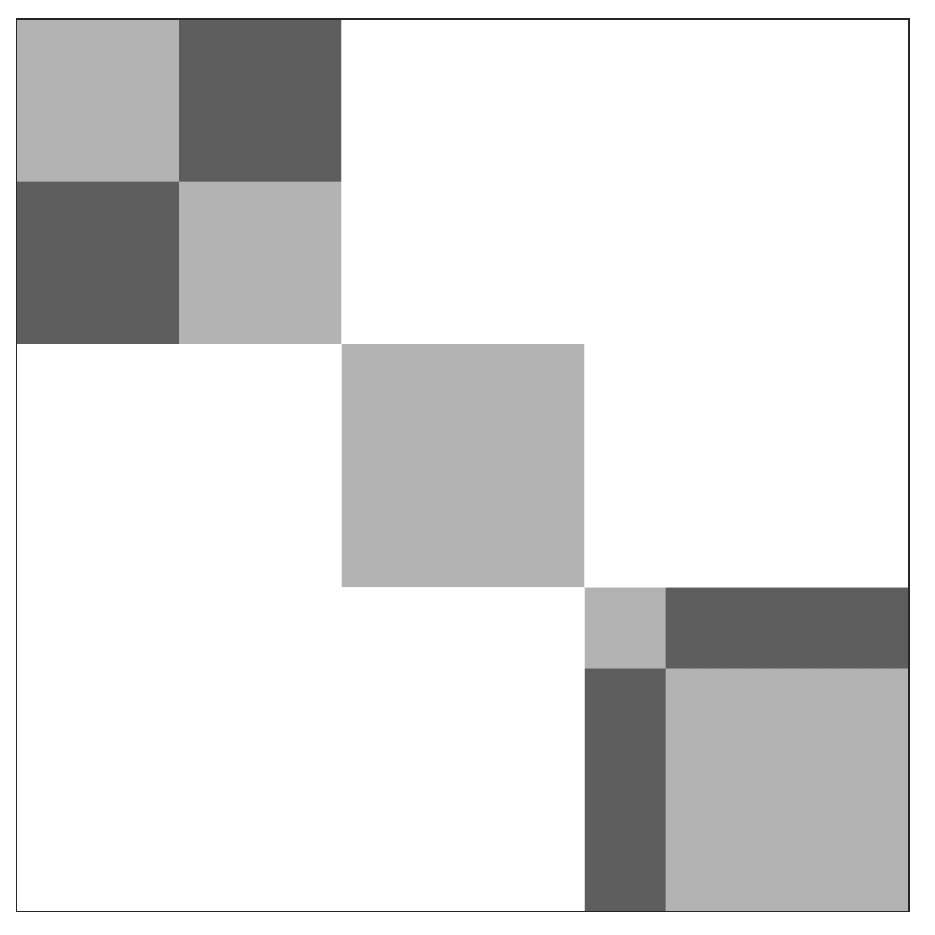}}\hfill
  \subfloat[$t=0$]{%
    \includegraphics[width=0.22\linewidth]{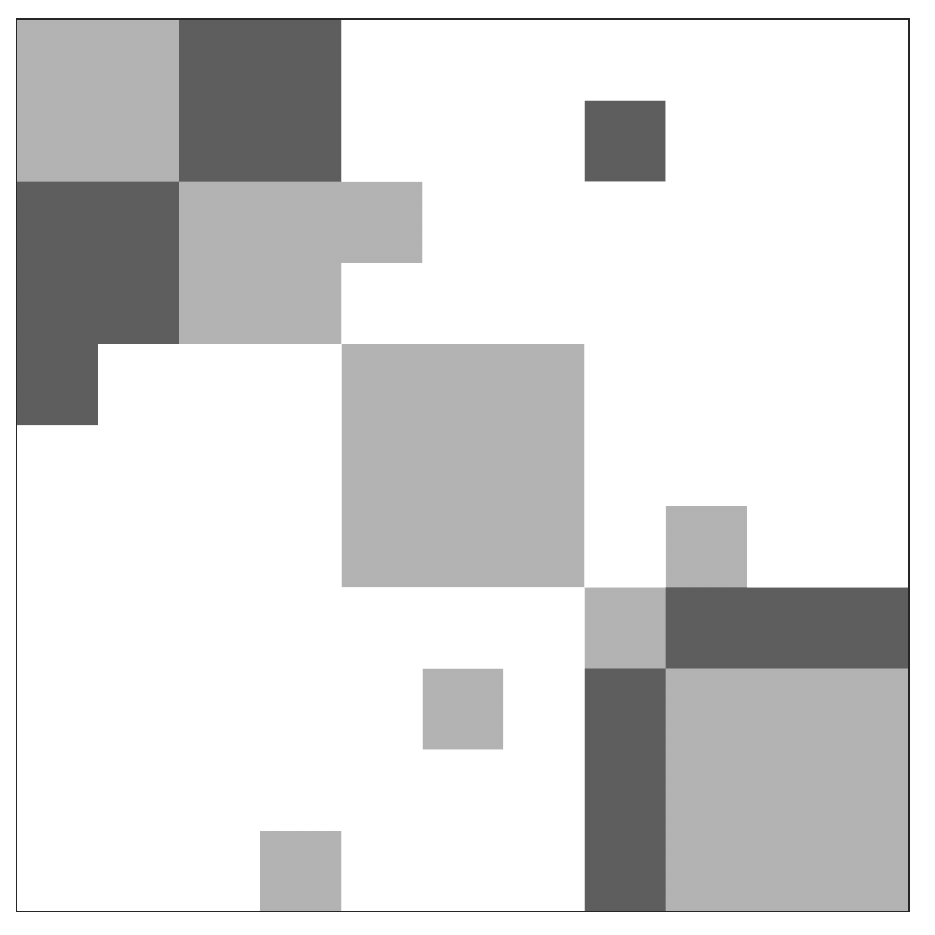}}\hfill
  \subfloat[$t=1$]{%
    \includegraphics[width=0.22\linewidth]{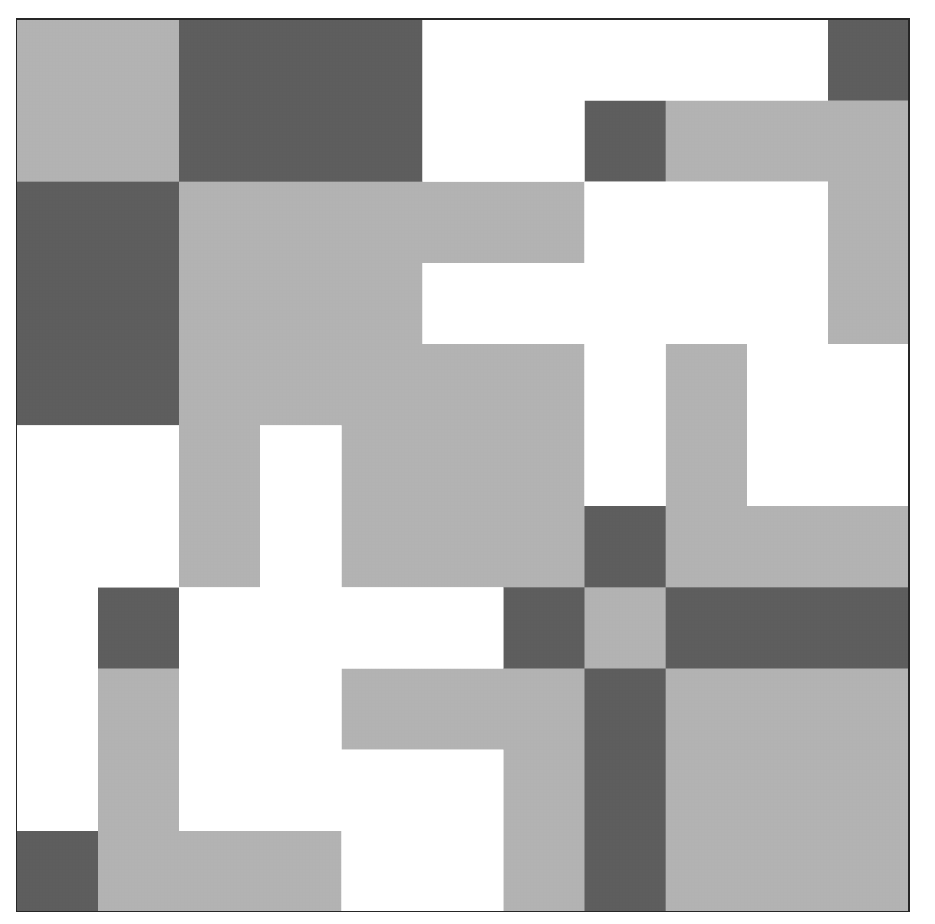}}\hfill
  \subfloat[$t=5$]{%
    \includegraphics[width=0.22\linewidth]{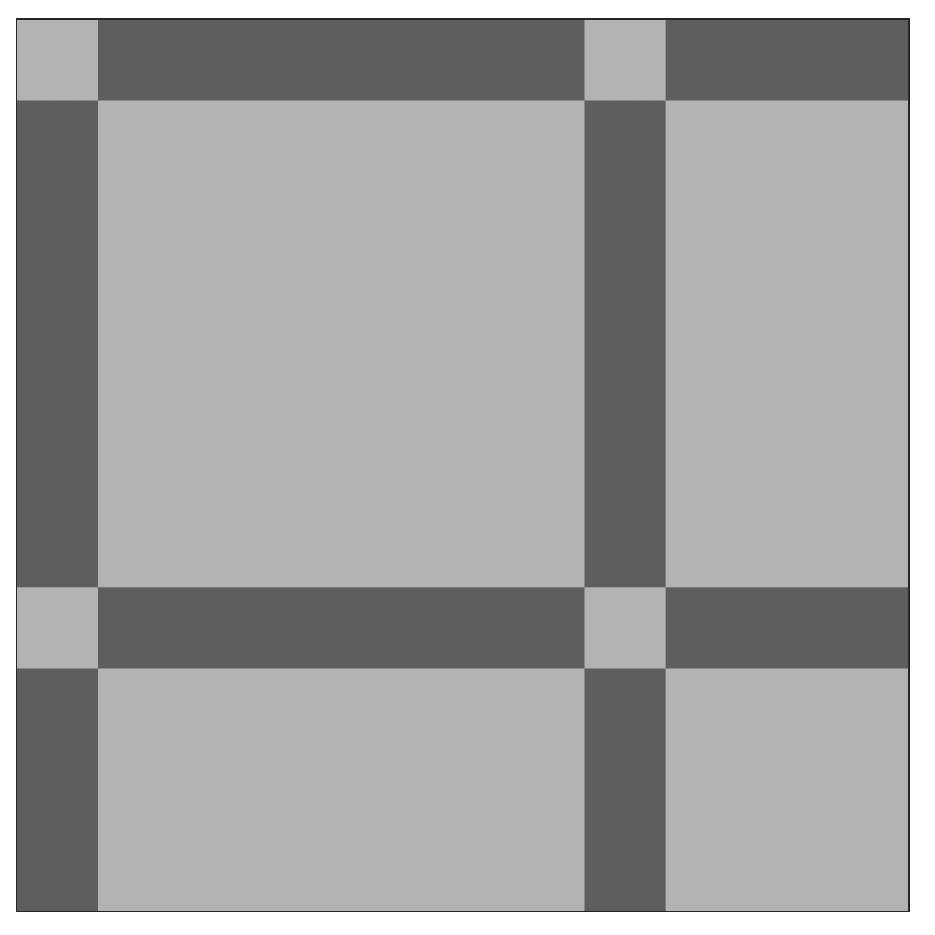}}
  \caption{Visualization of the evolution of the
    appraisal matrix under perturbations in the homophily-based model. For each entry, the dark
    gray color indicates a negative appraisal, while the light gray 
    indicates a positive one. The white color indicate no
    appraisal. The appraisal network has 11 nodes and is initially 
    in a social balance state with 3 isolated subgraphs.  With 6 links (4 positive and 2 negative links)
    added to the network, the appraisal network evolves to a
    single-clique structurally balanced state after 5 iterations.}
  \label{fig:eg- unstable-under=perturbation} 
\end{figure}
 
For the homophily model, extensive simulation observations indicate that social balance with $k>1$ isolated subgraphs is unstable under perturbations. With some links added to the network, the subnetworks connected by the added links merge into larger subnetworks and the perturbed network converges to another balanced state with fewer isolated subgraphs, see Fig.~\ref{fig:eg- unstable-under=perturbation} as an
example and the following two insightful scenarios.

\emph{Example 1: (Globalization of local conflicts)} Consider the appraisal network with two isolated subgraphs. Each subgraph is structurally balanced and made up of two antagonistic factions. The two factions in subgraph 1 are node sets $V_1$ and $V_2$ respectively, while the two factions in subgraph 2 are $V_3$ and $V_4$ respectively. Suppose one link with weight $\eta$ is added from one node in $V_1$ to one node in $V_3$. By computing the closed form expression of $X(2)$, we obtain that the perturbed appraisal network always recovers to a complete and structurally balanced network composed of two antagonistic factions. Moreover, if $\eta>0$, then the two factions are $V_1\cup V_3$ and $V_2\cup V_4$; If $\eta<0$, then the two factions are $V_1\cup V_4$ and $V_2\cup V_3$.
Figure~\ref{fig:perturbation-eg1} visualizes the behavior described above. In reality, such behavior could be interpreted as the escalation of local conflicts.  
One real example of such phenomena is the formation of the globalized conflicts between the Axis and the Ally in World War II, after the Nazi German allied with the Imperial Japan.

\begin{figure}
\begin{center}
\subfloat[$\eta>0$]{\label{fig:perturbation-eg1-ep-pos} \includegraphics[width=.35\linewidth]{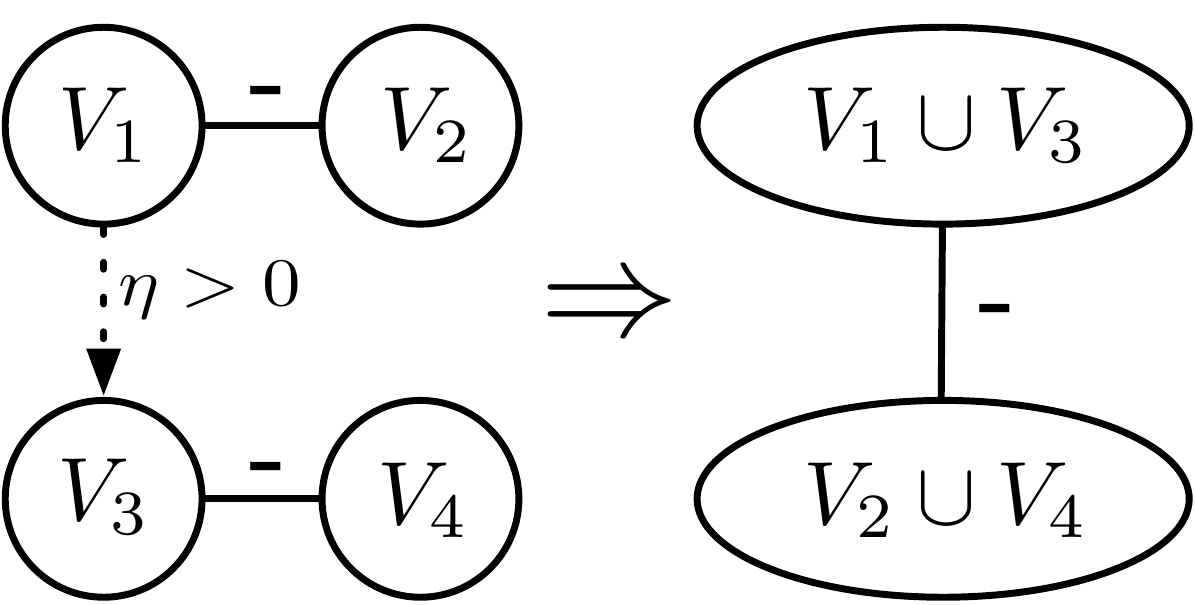}}\qquad
\subfloat[$\eta<0$]{\label{fig:perturbation-eg1-ep-neg} \includegraphics[width=.35\linewidth]{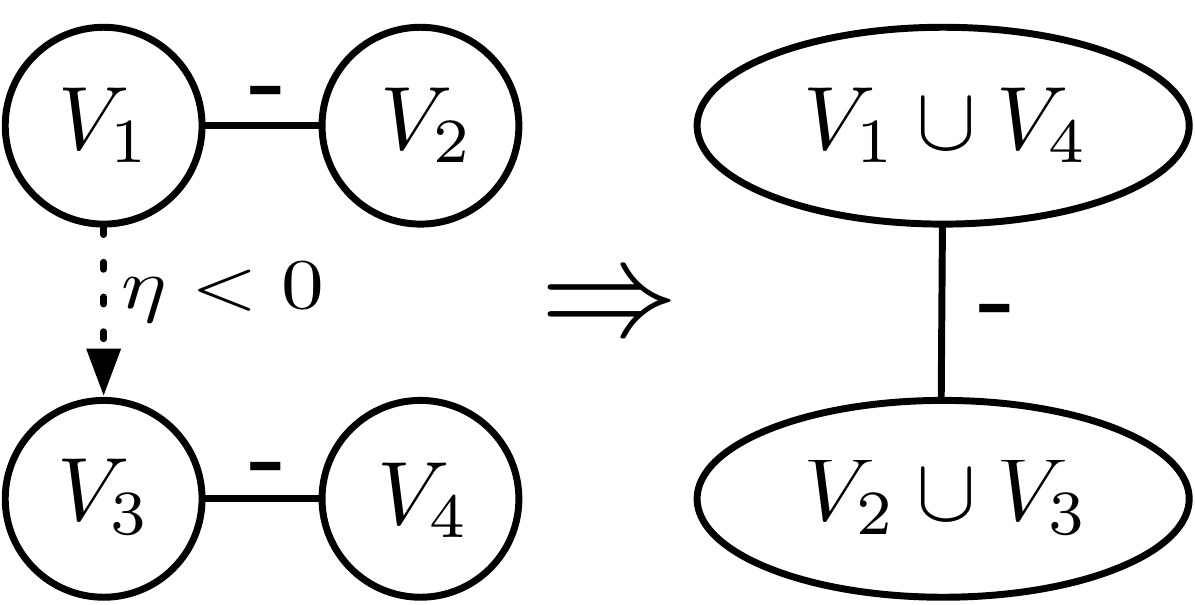}}
\caption{\footnotesize Visual illustration of the behavior of the social balance with $2$ isolated subgraphs with the addition of one inter-subgraph link.}\label{fig:perturbation-eg1}
\end{center}
\end{figure}

\emph{Example 2: (Competition for ally and mediation of conflicts)}
Consider an appraisal network with two isolated subgraphs: subgraph
1 with two antagonistic factions $V_1 =
\{1,\dots,n_1\}$ and $V_2=\{n_1+1,\dots, n_1+n_2\}$, and subgraph 2
with only one faction $V_3 = \{n_1+n_2+1,\dots,
n_1+n_2+n_3\}$. Suppose the appraisal matrix associated with subgraph 1
is given by $\alpha bb^{\top}$, where
$b=(\vect{1}_{n_1}^{\top},-\vect{1}_{n_2}^{\top})^{\top}$, and
$\alpha>0$ represents the sentiment strength inside subgraph
1. Similarly, the appraisal matrix associated with subgraph 2 is given
by $\hat{\alpha}\hat{b}\hat{b}^{\top}$, where $\hat{b}=\vect{1}_{n_3}$
and $\hat{\alpha}>0$ represents the sentiment strength inside subgraph
2. Imagine then that both $V_1$ and $V_2$ aim to ally
with $V_3$. Accordingly, suppose that, in order to ally with $V_3$, each node in $V_1$ builds a
bilateral link with each node in $V_3$, with link weight
$\epsilon_1>0$, while each node in $V_2$ builds a bilateral link with
each node in $V_3$ with weight $\epsilon_2>0$. With all these links
added, the associated appraisal matrix takes the following form:
\begin{equation*}
X(0) = 
\begin{bmatrix}
\alpha \vect{1}_{n_1} \vect{1}_{n_1}^{\top} & -\alpha \vect{1}_{n_1}\vect{1}_{n_2}^{\top} & \epsilon_1 \vect{1}_{n_1}\vect{1}_{n_3}^{\top} \\
\smallskip
-\alpha \vect{1}_{n_2}\vect{1}_{n_1}^{\top} & \alpha \vect{1}_{n_2}\vect{1}_{n_2}^{\top} & \epsilon_2 \vect{1}_{n_2}\vect{1}_{n_3}^{\top} \\
\smallskip
\epsilon_1 \vect{1}_{n_3}\vect{1}_{n_1} & \epsilon_2 \vect{1}_{n_3}\vect{1}_{n_1}^{\top} & \hat{\alpha} \vect{1}_{n_3}\vect{1}_{n_3}^{\top}
\end{bmatrix}.
\end{equation*}  
Along the evolution of $X(t)$ determined by $X(0)$,
we obtain the following numerical results.

(i) If $\epsilon_1 n_1 > \epsilon_2 n_2$, i.e., faction $V_1$ takes greater effort than $V_2$ in allying with $V_3$, then $V_1$ gains at least one ally, either $V_2$ or $V_3$. Moreover, the following conditions $\epsilon_1 n_1 - \epsilon_2 n_2  \ge \hat{\alpha}\epsilon_2 n_3/\alpha$ and $\epsilon_1 \epsilon_2 n_3 \le \alpha^2 (n_1+n_2)$
guarantee that $V_1$ ally with $V_3$; This statement also holds when all the subscripts $1$ and $2$ are switched;

(ii) If $\epsilon_1 \epsilon_2 n_3 \le \alpha^2 (n_1+n_2)$, then $V_3$ eventually gains at least one ally. That is, $V_3$ avoids the situation in which $V_1$ and $V_2$ end up allying with each other against $V_3$;

(iii) Any of the following conditions guarantees that no negative link exists in the asymptotic appraisal network: (1) $\epsilon_1 \epsilon_2 n_3 \ge \alpha^2 (n_1+n_2)$ and $\epsilon_1 n_1 - \epsilon_2 n_2=0$; (2) $\epsilon_1 \epsilon_2 n_3 \ge \alpha^2 (n_1+n_2)$ and $0<\epsilon_1 n_1 -\epsilon_2 n_2 \le \epsilon_2 \hat{\alpha}n_3$; (3) $\epsilon_1 \epsilon_2 n_3 \ge \alpha^2 (n_1+n_2)$ and $0<\epsilon_2 n_2 - \epsilon_1 n_1 \le \epsilon_1 \hat{\alpha}n_3$. Notice that the inequality
$\epsilon_1 \epsilon_2 n_3 \ge \alpha^2 (n_1+n_2)$
is required for all the three sufficient conditions. The right-hand side of this inequality above reflects the ``scale'' of the conflicts between factions $V_1$ and $V_2$, while the left-hand side is $V_1$ and $V_2$'s average efforts in allying with $V_3$, multiplied by the size of $V_3$. From the three sufficient conditions, we learn that, the larger the size of $V_3$, the more capable it is of mediating the conflicts between $V_1$ and $V_2$. In addition, $V_1$ and $V_2$'s strong willingness to ally with $V_3$, as well as the sentiment strength inside $V_3$, i.e., $\hat{\alpha}$, also help mediate the conflicts.

\subsection{Distribution of initial conditions and formation of factions in the homophily-based model}
\label{distInit}

We investigate numerically, for the homophily-based model, how initial appraisal distribution determines whether the appraisal network evolves to only one faction or two antagonistic factions. We randomly and independently sample the entries of $X(0)$ from the uniform distribution on $[\xmin,\xmax]$, for which $\ave(\xmin,\xmax)=(\xmax+\xmin)/2$ indicates how the initial appraisals are biased towards being positive. We set $\xmax-\xmin=2$ and change the values of $\ave(\xmin,\xmax)$ and the number of agents. Given $[\xmin,\xmax]$, 30 samples of the initial condition $X(0)$ are independently randomly generated  and for each $X(0)$ we count how many factions appear at $X(500)$. Since any $X(0)$ and $-X(0)$ lead to the same $X(1)$ and $X(t)$ thereafter, we only consider different values of $\ave(\xmin,\xmax)\geq 0$. Figure~\ref{fig_lscale} shows that, for fixed network size, the
smaller the value of $\ave(\xmin,\xmax)$, the more likely it is to find two antagonistic factions; for fixed value of $\ave(\xmin,\xmax)$, the larger the network size, the more likely that only one faction emerges.

Note that similar numerical study in~\cite{SAM-JK-RDK-SHS:11} for the continuous-time influence-like model indicates that, the appraisal network evolves to two antagonistic factions if the initial mean appraisal is non-positive. The appraisal network evolves to all-friendly state if the initial mean is positive. However, such results in~\cite{SAM-JK-RDK-SHS:11} only hold for the limit case of infinitely large network size $n$.

\begin{figure}[!t] 
    \centering

{
        \includegraphics[width=0.7\linewidth]{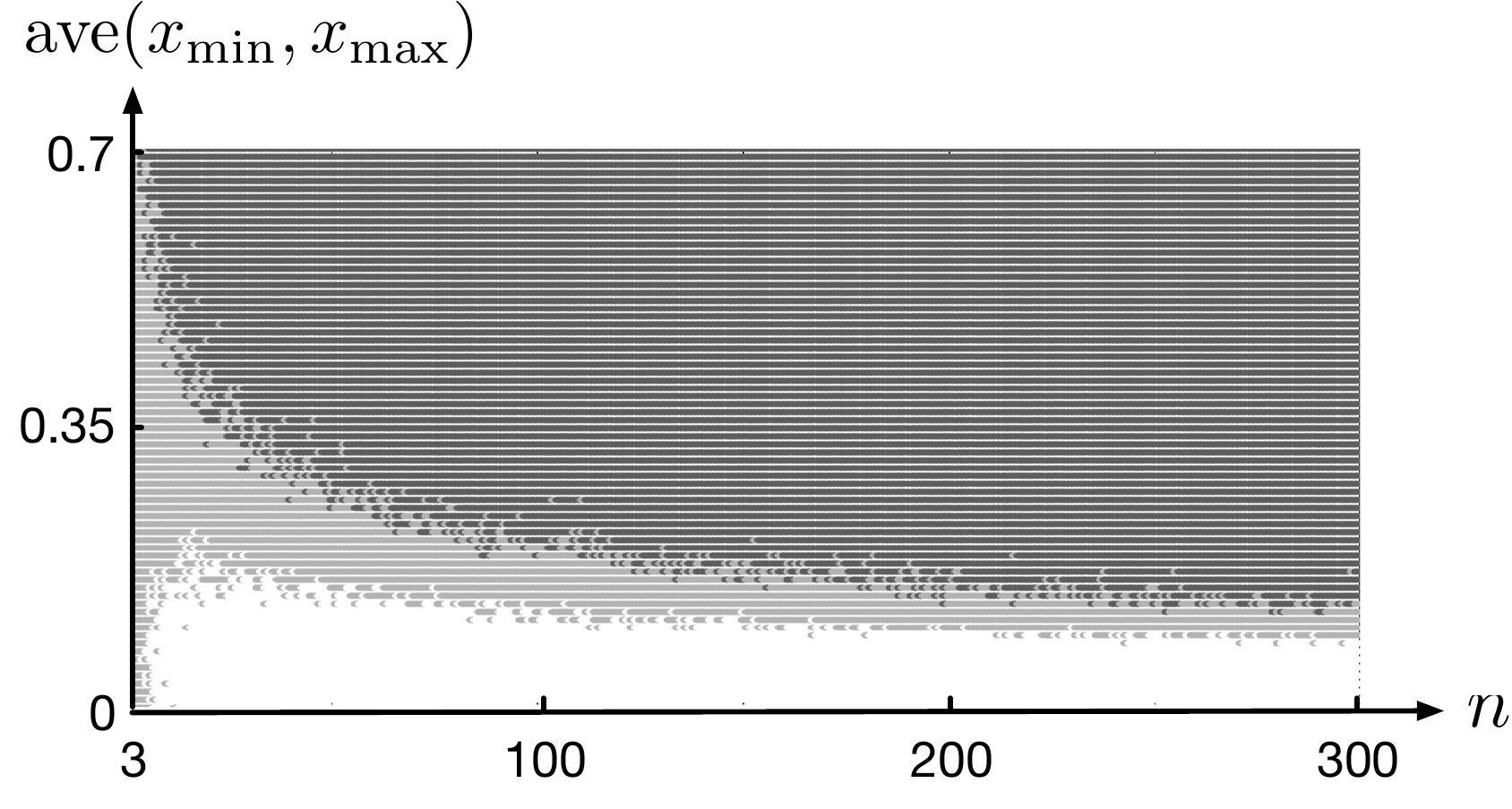}}
    \label{fig_lscalea} 
  \caption{\footnotesize Formation of factions under different initial
    condition distributions for the homophily-based model.  
    The white color indicates the presence of two factions
    in all the 30 random samples, while the dark gray color indicates the
    presence of one faction in all of the samples. The light gray 
    color indicates any other case.}
  \label{fig_lscale} 
\end{figure}

\section{Conclusion}
\label{conclSec}
This paper proposes both homophily-based and influence-based discrete-time models for the
bounded evolution of interpersonal appraisal networks towards social
balance. For either model, the set of fixed points include all the possible balanced configurations, in the sense of sign pattern, of the appraisal network. Under the non-vanishing appraisal condition, we prove that
both models exhibit asymptotic convergence to structurally balanced
networks, while the convergence property holds for larger initial conditions set in the homophily-based model than in the influence-based model. Moreover, our models admits the existence of multiple isolated subgraphs in the final structure of the evolved appraisal network. Numerical
study indicates how the final emergence of factions in the social
network is sensitive to the initial appraisal distribution, and how the system transits from one fixed point to another under perturbations. 

We remark that our models and the previous continuous-time models~\cite{KK-PG-PG:05,SAM-JK-RDK-SHS:11,VAT-PVD-PDL:13,PJ-NEF-FB:13n} all adopt the definition of social balance for complete graphs, or isolated complete subgraphs in our paper, which implies that individuals interact with everyone in the group/subgroup. This assumption limit the scope of the application of our models to (groups of) small-size groups, which are usually assumed to be complete graphs.

Possible future
research directions include a better understanding of the
influence-based model for arbitrary initial conditions, a validation
of the proposed models with laboratory and/or field data, the study of
asynchronous models with pairwise updates, and further study of conditions
and cases in which one socio-psychological mechanism dominates the
other.

\bibliographystyle{plainurl}
\bibliography{alias,Main,FB}

\end{document}